\documentclass[11 pt]{article}
\usepackage{amsfonts,amsmath,color,amsthm,amssymb, url, enumerate, bbm, subfig}
\usepackage{graphicx}
\usepackage[DIV=14,BCOR=2mm,headinclude=true,footinclude=false]{typearea}
\usepackage[font=small,labelfont=bf]{caption}
\usepackage{hyperref}
\usepackage{tikz, etoolbox}
\usetikzlibrary{shapes}
\usetikzlibrary{arrows}
\usepackage[margin=1in]{geometry}

\newcommand{\eps}{\epsilon}

\newcommand{\al}{\alpha}

\newcommand{\f}{\frac}

\newcommand{\tr}{\text{trace}}

\renewcommand{\th}{\theta}

\newcommand{\ol}{\overline}

\newcommand{\rar}{\rightarrow}

\newcommand{\N}{\mathbb{N}}
\newcommand{\R}{\mathbb{R}}

\newcommand{\bi}{\mathbbm{1}}

\definecolor{purple}{RGB}{128,0,128}

\def\qed{\vbox{\hrule\hbox{\vrule\kern3pt\vbox{\kern6pt}\kern3pt\vrule}\hrule}}

\def\eps{\epsilon}

\newcommand{\SDPOPT}{\text{OPT}_{\text{SDP}}}
\newcommand{\TSPOPT}{\text{OPT}_{\text{TSP}}}

\newtheorem{thm}{Theorem}[section]
\newtheorem{lm}[thm]{Lemma}
\newtheorem{cor}[thm]{Corollary}

\newtheorem{prop}[thm]{Proposition}

\newtheorem{cm}[thm]{Claim}

\newtheorem*{thm*}{Theorem}
\newtheorem*{lm*}{Lemma}
\newtheorem*{prop*}{Proposition}
\newtheorem{ques}[thm]{Question}

\theoremstyle{definition}
\newtheorem{exam}[thm]{Example}

\newcommand{\footremember}[2]{%
    \footnote{#2}
    \newcounter{#1}
    \setcounter{#1}{\value{footnote}}%
}
\newcommand{\footrecall}[1]{%
    \footnotemark[\value{#1}]%
    }

\newtheoremstyle%
 {Aside}%
 {}{}%
 {\color{purple}\itshape}
 {}%
 {\color{purple}\bfseries}%
 {\color{purple}.}%
 { }{}
\theoremstyle{Aside}

\makeatletter
\newcommand{\clonelabel}[2]{\@bsphack
  \expandafter\ifx\csname r@#2\endcsname\relax
  \else\protected@write\@auxout{}{\string\newlabel{#1}%
    {\csname r@#2\endcsname}}%
  \fi
  \expandafter\ifx\csname r@#2@cref\endcsname\relax
  \else\protected@write\@auxout{}{\string\newlabel{#1@cref}%
    {\csname r@#2@cref\endcsname}}%
  \fi
  \@esphack}
\makeatother

\title{Semidefinite Programming Relaxations of the Traveling Salesman Problem and Their Integrality Gaps}
\author{Samuel C. Gutekunst\footremember{Cornell}{Operations Research and Information Engineering, Cornell University} and David P. Williamson\footrecall{Cornell} }
\date{}

\begin{document}
\maketitle

\begin{abstract}
The traveling salesman problem (TSP) is a fundamental problem in combinatorial optimization.  Several semidefinite programming relaxations have been proposed recently that exploit a variety of mathematical structures including, e.g., algebraic connectivity, permutation matrices, and association schemes.  The main results of this paper are twofold.  First,  de Klerk and Sotirov \cite{Klerk12b} present an SDP based on permutation matrices and symmetry reduction; they show that it is incomparable to the subtour elimination linear program, but generally dominates it on small instances.  We provide a family of \emph{simplicial TSP instances} that shows that the integrality gap of this SDP is   unbounded.  Second, we show that these simplicial TSP instances imply  the unbounded integrality gap of every SDP relaxation of the TSP mentioned in the survey on SDP relaxations of the TSP in Section 2 of Sotirov \cite{Sot12}.  In contrast, the subtour LP performs perfectly on simplicial instances.  The simplicial instances thus form a natural litmus test for future SDP relaxations of the TSP.
\end{abstract}

\section{Introduction} 

In this paper, we consider a relaxation of the traveling salesman problem (TSP) based on semidefinite programs.  The TSP is a fundamental problem in combinatorial optimization, combinatorics, and theoretical computer science.  An input consists of a set $[n]:=\{1, 2, 3, ..., n\}$ of $n$ cities and, for each pair of cities $i, j \in [n]$, an associated cost or distance $d_{ij}\geq 0$ reflecting the cost or distance of traveling from city $i$ to city $j$.  Throughout this paper, we assume that the edge costs $d_{ij}$ are \emph{symmetric} (so that $d_{ij}=d_{ji}$ for all $i, j\in [n]$) and \emph{metric} (so that $d_{ij}\leq d_{ik}+d_{kj}$ for all $i, j, k\in [n]$).  The TSP is then to find a minimum-cost tour visiting each city exactly once.  Treating the cities as vertices of the complete, undirected graph $K_n,$ and treating an edge $\{i, j\}$ of $K_n$ as having cost $d_{ij}$, the TSP is equivalently to find a minimum-cost Hamiltonian cycle on $K_n.$

The TSP (with the implicit assumptions that the edge costs are metric and symmetric) is a canonical NP-hard problem; finding a polynomial-time approximation algorithm with as strong a performance guarantee as possible remains a major open question.  Currently it is known to be NP-hard to approximate TSP solutions in polynomial time to within any constant factor  $\alpha<\f{123}{122}$ (see Karpinski,  Lampis,  and Schmied \cite{Karp15}).  In contrast, the strongest positive performance guarantee dates back more than 40 years: the Christofides-Serdyukov  algorithm \cite{Chr76,Serd78} finds a Hamiltonian cycle in polynomial time that is at most  a factor of $\f{3}{2}$ away from the optimal TSP solution.

One powerful technique for analyzing TSP approximation algorithms is to \emph{relax} the discrete set of Hamiltonian cycles.  The  prototypical example is the subtour elimination linear program (also referred to as the Dantzig-Fulkerson-Johnson relaxation \cite{Dan54} and the Held-Karp bound \cite{Held70}, and which we will refer to as the subtour LP).  The subtour LP is a \emph{relaxation} of the TSP because 1) every Hamiltonian cycle has a corresponding feasible solution to the subtour LP, and 2) the value of the subtour LP for such a feasible solution equals the cost of the corresponding Hamiltonian cycle.  As a result, the optimal value of the subtour LP is a lower bound on the optimal solution to the TSP.  Wolsey \cite{Wol80},   Cunningham \cite{Cun86}, and Shmoys and Williamson \cite{Shm90}
 show that the Christofides-Serdyukov algorithm produces a (not-necessarily optimal) Hamiltonian cycle that is within a factor of $\f{3}{2}$ of the optimal value of the subtour LP.  Combining these two observations shows that the Christofides-Serdyukov algorithm satisfies the following chain of inequalities.
\begin{align*}
\text{Optimal TSP solution } &\leq \text{ Cost of cycle produced by Christofides-Serdyukov algorithm}\\
&\leq \f{3}{2} \text{ Optimal value of subtour LP} \\
&\leq \f{3}{2} \text{ Optimal TSP solution.}
\end{align*}
Hence, the Christofides-Serdyukov algorithm is a $\f{3}{2}$-approximation algorithm for the TSP.  Moreover, the \emph{integrality gap} of the subtour LP, which measures the worst-case performance of a relaxation relative to the TSP, is at most $\f{3}{2}$: for any instance, the ratio of the optimal TSP solution to the optimal value of the subtour LP cannot be more than $\f{3}{2}$.  Note that, if the subtour LP did not have a constant-factor integrality gap, it would not be possible to use the LP as above to show that a TSP algorithm was a constant-factor approximation algorithm.  Goemans \cite{Goe95b} conjectured that the integrality gap of the subtour LP is $\f{4}{3},$ though the $\f{3}{2}$ bound of  Wolsey \cite{Wol80},   Cunningham \cite{Cun86}, and Shmoys and Williamson \cite{Shm90} remains state-of-the-art.

More recently, several TSP relaxations based on semidefinite programs (SDPs) have been proposed; see Section 2 of Sotirov \cite{Sot12} for a short  survey.  Cvetkovi{\'c}, {\v{C}}angalovi{\'c}, and Kova{\v{c}}evi{\'c}-Vuj{\v{c}}i{\'c} \cite{Cvet99} gave a relaxation based on adjacency matrices and algebraic connectivity.  De Klerk, Pasechnik, and Sotirov \cite{Klerk08} introduced a relaxation based on the theory of association schemes  (see also de Klerk, de Oliveira Filho, and Pasechnik \cite{Klerk12}). Zhao,  Karisch, Rendl, 
and Wolkowicz \cite{Zhao98} introduce a relaxation to the more general Quadratic Assignment Problem (QAP), a special case of which is the TSP.  Their relaxation is based on properties of permutation matrices; de Klerk et al.\ \cite{Klerk08} show the optimal value of their SDP coincides with the optimal value of the SDP introduced by Zhao et al.\ \cite{Zhao98} when specialized to the TSP.   Sotirov \cite{Sot12} summarizes two equivalent interpretations of this latter SDP relaxation of the QAP: First, it is equivalent to a similar SDP relaxation of the QAP also based on permutation matrices from Povh and Rendl \cite{Povh09} (with equivalence shown in  Povh and Rendl \cite{Povh09}).  Second, it is equivalent to applying the $N^+$ lift-and-project operator of Lov{\'a}sz and Schrijver \cite{Lov91} to a QAP polytope; this equivalence is shown in Burer and Vandenbussche \cite{Bur06} and Povh and Rendl \cite{Povh09}.   
Anstreicher \cite{Ans00} gives  another SDP relaxation of the QAP.  When specialized to the TSP, it is equivalent to the projected eigenvalue bound of Hadley,  Rendl, and Wolkowicz \cite{Had92}. 

Most recently, de Klerk and Sotirov \cite{Klerk12b} apply symmetry reduction to strengthen the QAP relaxation of Povh and Rendl \cite{Povh09} in certain cases.  This strengthened QAP relaxation can be applied to the TSP and  de Klerk and Sotirov \cite{Klerk12b} evaluate the strengthened QAP relaxation on the 24 classes of facet defining inequalities for the TSP on 8 vertices.  While solving the SDP is computationally demanding, their results are promising: the strengthened QAP performs at least as well as the subtour LP on all but one of the 24 instances and generally outperforms the subtour LP.

Although computationally involved, these SDPs are based on a broad variety of rich combinatorial structures which has led to several theoretical results. Goemans and Rendl \cite{Goe00} show that the SDP relaxation of Cvetkovi{\'c} et al.\ \cite{Cvet99} is weaker than the subtour LP in the following sense: Any solution to the subtour LP implies an equivalent feasible solution for the SDP of Cvetkovi{\'c} et al.\ of the same cost.  Both optimization problems are minimization problems and the SDP is optimizing over a broader search set, so the optimal  value for the SDP of Cvetkovi{\'c} et al.\ cannot be closer than the optimal value  of the subtour LP to the optimal TSP cost.    However, de Klerk et al.\ \cite{Klerk08} show the exciting result that their SDP is incomparable with the subtour LP: there are instances where the optimal value of their SDP is closer to the optimal TSP cost than the optimal value of the subtour LP, and vice versa.  Moreover, de Klerk et al.\ \cite{Klerk08} show that their SDP is stronger than the earlier SDP of Cvetkovi{\'c} et al.\  \cite{Cvet99}: any feasible solution for the SDP of  de Klerk et al.\ \cite{Klerk08} implies a feasible solution for the SDP of Cvetkovi{\'c} et al.\  \cite{Cvet99} of the same cost.

Gutekunst and Williamson \cite{Gut17} show that the SDP relaxations of both  Cvetkovi{\'c} et al.\ \cite{Cvet99} and de Klerk et al.\ \cite{Klerk08}, however, have unbounded integrality gaps.  Moreover, they have a counterintuitive non-monotonicity property: in certain instances it is possible to artificially add vertices (in a way that preserves metric and symmetric edge costs) and arbitrarily lower the cost of the optimal solution to the SDP.  Such a property contrasts with both the TSP and subtour LP, which are known to be monotonic (see Section \ref{UBGap}).

The main results of this paper are to complete the characterization of integrality gaps of every SDP relaxation of the TSP mentioned in  Sotirov \cite{Sot12} and to introduce a family of instances that implies every such SDP has an unbounded integrality gap and is non-monotonic.  To do so, we show that the SDP of de Klerk and Sotirov \cite{Klerk12b}  has an unbounded integrality gap (and in turn, has the same non-monotonicity property of Cvetkovi{\'c} et al.\ \cite{Cvet99} and de Klerk et al.\ \cite{Klerk08}).  Doing so further implies that no SDP relaxation of the TSP surveyed in Sotirov \cite{Sot12}  can be used in proving approximation guarantees on TSP algorithms in the same way as the subtour LP.   The family of instances we use generalizes those from
Gutekunst and Williamson \cite{Gut17} to a new family of TSP instances which we call \emph{simplicial TSP instances}, as they can be viewed as placing groups of vertices at the extreme points of a simplex.  This family forms an intriguing set of test instances for SDP relaxations of the TSP: the vertices of the TSP instance can be embedded into $\R^d$ (for a $d$ that grows as the integrality gap increases), the integrality gap of the subtour LP on these instances is 1 (i.e. the optimal value of the subtour LP on any instance in this family matches the cost of the TSP solution), but these instances imply an unbounded integrality gap for at least the following SDPs:
\begin{itemize}
\item The SDP TSP relaxation of Cvetkovi{\'c}, {\v{C}}angalovi{\'c}, and Kova{\v{c}}evi{\'c}-Vuj{\v{c}}i{\'c} \cite{Cvet99} (based on algebraic connectivity). 
\item The SDP TSP relaxation of de Klerk, Pasechnik, and Sotirov \cite{Klerk08} based on the theory of association schemes  (see also de Klerk, de Oliveira Filho, and Pasechnik \cite{Klerk12}). 
\item The SDP QAP relaxation of Zhao,  Karisch, Rendl, and Wolkowicz \cite{Zhao98}, when specialized to the TSP (based on permutation matrices, and shown by de Klerk et al.\ \cite{Klerk08} to have an optimal value coinciding with the SDP of de Klerk et al.\ \cite{Klerk08}).
\item The SDP QAP relaxation of Povh and Rendl \cite{Povh09}, when specialized to the TSP (based on permutation matrices, and shown by Povh and Rendl \cite{Povh09} to be equivalent to the SDP of  Zhao et al.\ \cite{Zhao98}).
\item The SDP QAP relaxation of  de Klerk and Sotirov \cite{Klerk12b}, when specialized to the TSP (obtained by performing symmetry reduction on the SDP of  de Klerk et al.\ \cite{Klerk08}).
\item The SDP QAP relaxation of Anstreicher \cite{Ans00}, when specialized to the TSP (equivalent to the projected eigenvalue bound of Hadley,  Rendl, and Wolkowicz \cite{Had92}).
\end{itemize}

In Section \ref{SDPTSP}, we introduce the notation we will use and provide background on the SDP of de Klerk and Sotirov \cite{Klerk12b}.  In Section \ref{Idea}, we show how the instances of Gutekunst and Williamson \cite{Gut17} directly imply that the integrality gap of the SDP of  Povh and Rendl \cite{Povh09} is unbounded, but only that the integrality gap of  the SDP of de Klerk and Sotirov \cite{Klerk12b} is at least 2.  This result motivates the generalized simplicial instances we formalize in Section \ref{UBGap}. In Section \ref{UBGap}, we also prove our main result.  We specifically show that for $z\in\N,$ the simplicial instances in $\R^{2z-1}$ imply an integrality gap for the SDP  of de Klerk and Sotirov \cite{Klerk12b} of at least $z$.  We do so by finding a family of instances where the SDP cost can be bounded by $2+\eps$ for any $\eps>0$ (with sufficiently large $n$), while the TSP cost grows arbitrarily.   As a corollary, we show that the SDP of de Klerk and Sotirov \cite{Klerk12b} is again non-monotonic.  We conclude in Section \ref{sec:conc} by discussing two open questions about SDP-based relaxations of the TSP.

\section{SDP Relaxations of the TSP}\label{SDPTSP}
\subsection{Notation and Preliminaries}

Throughout this paper, we use $J_m$ and $I_m$ to respectively denote the all-ones and identity matrix in $\R^{m\times m}$.   We let $e^{(m)}_i$ denote the $i$th standard basis vector in $\R^m$ and let $e^{(m)}:=e_1^{(m)}+\cdots + e_m^{(m)}$ denote the all-ones vector in $\R^m.$  We let $E_{ij}^{(m)}:=e_i^{(m)}\left(e_j^{(m)}\right)^T$ denote the $m\times m$ matrix with a one in the $i, j$th position and zeros elsewhere.

We let $\mathbb{S}^{m\times m}$ denote the set of real, symmetric matrices in $\R^{m\times m}$ and let $\Pi_m$ be the set of $m\times m$ permutation matrices. $Y\succeq 0$ denotes that $Y$ is a positive semidefinite matrix; for $Y\in \mathbb{S}^{m\times m},$ $Y\succeq 0$ means that all eigenvalues of $Y$ are nonnegative.  $Y\geq 0$ denotes that $Y$ is a nonnegative matrix entrywise. 

We will use several matrix operations from linear algebra.  For a matrix $M\in \R^{m\times m}$ and  $S_1, S_2 \subset[m],$ let $M[S_1, S_2]$ denote the submatrix of $M$ with rows in $S_1$ and columns in $S_2.$  When $S_1=S_2,$ we simplify notation and write $M[S_1]:=M[S_1, S_1].$  For a vector $x\in \R^m$, let $Diag(x)$ be the $m\times m$ diagonal matrix whose $i, i$-th entry is $x_i.$  For a matrix $Y$, let $\tr(Y)$ denote the trace of $Y$, i.e., the sum of its diagonal entries.  For $A, B \in  \mathbb{S}^{m\times m}$, note that $$\tr(AB)=\sum_{i=1}^m \sum_{j=1}^m A_{ij}B_{ij}=\langle A, B\rangle,$$ the matrix inner product.  For an $m\times m$ matrix $Y$, let $vec(Y)$ be the vector in $\R^{m^2}$ that stacks the columns of $Y$.  Finally, for matrices $A, B$ of arbitrary dimension, $A\otimes B$ denotes the Kronecker product of $A$ and $B$.  The Kronecker product has particularly nice spectral properties. If $A\in \R^{a\times a}$ and $B\in \R^{b \times b}$  have respective eigenvalues $\lambda_i(A)$ and $\lambda_j(B)$ for $i=1, ..., a$ and $j=1, ..., b$, the $ab$ eigenvalues of $A\otimes B$ are the $ab$ products $\lambda_i(A)\lambda_j(B).$ See, e.g., Theorem 4.2.12 in Chapter 4 of Horn and Johnson  \cite{Hor91}.

We will regularly work with \emph{circulant} matrices.  A circulant matrix in $\R^{m\times m}$  has the form $$\begin{pmatrix} c_0 & c_1 & c_2 & c_3 & \cdots & c_{m-1} \\ c_{m-1} & c_0 & c_1 & c_2 & \cdots & c_{m-2} \\ c_{m-2} & c_{m-1} & c_0 & c_1 & \ddots & c_{m-3} \\ \vdots & \vdots & \vdots & \vdots & \ddots & \vdots \\ c_1 & c_2 & c_3 & c_4 & \cdots & c_0 \end{pmatrix} = \left(c_{(t-s) \text{ mod } m}\right)_{s, t=1}^m.$$ Such a matrix is symmetric if $c_i=c_{m-i}$ for $i=1, ..., m-1$.   We use a standard basis of symmetric circulant matrices in $\R^{m\times m}$ consisting of matrices $C_0^{(m)}, C_1^{(m)}, ..., C_d^{(m)}$ where, for $i=1, ..., d-1$, $C_i^{(m)}$ is the symmetric circulant $m\times m$ matrix with $c_i=c_{m-i}=1$ and $c_j=0$ otherwise.  We set $C_0^{m}=2I$ and, when $m$ is even, set $C_{m/2}^{(m)}$ to be the matrix where $c_{m/2}=2$ and $c_j=0$ otherwise.  Note that, following these definitions, each $C_i^{(m)}$ has all rows sum to 2.  When clear from context, we will suppress the dependence on the dimension and use, e.g., $C_i$ rather than $C_i^{(m)}.$ We use $\mathcal{A}(G)$ to denote the adjacency matrix of a graph $G$ and $\mathcal{C}_m$ to denote the cycle graph on $m$ vertices in lexicographic order.  Note that $\mathcal{A}(\mathcal{C}_m)=C_1^{(m)}.$

Throughout the remainder of this paper we will take $n$ to be the number of cities/vertices of a TSP instance.  We will assume that $n$ is even and let $d= \f{n}{2}.$  We reserve $D$ as the matrix of edge costs or distances (so that for $1\leq i\leq n$ and $1\leq j\leq n$, $D_{ii}=0$ and $D_{ij}=d_{ij}$ is the cost of traveling between cities $i$ and $j$).  We implicitly assume that the edge costs $d_{ij}$ defining $D$ are symmetric and metric. 

Let $\SDPOPT(D)$ and $\TSPOPT(D)$ respectively denote the optimal value to an SDP relaxation and the cost of an optimal TSP solution for a given matrix of costs $D.$   If $\mathcal{D}$ is the set of all cost matrices corresponding to metric and symmetric TSP instances, the \emph{integrality gap} of the SDP is $$\sup_{D\in\mathcal{D}} \f{\TSPOPT(D)}{\SDPOPT(D)}.$$ This ratio is bounded below by 1 for any SDP that is a relaxation of the TSP (as the optimal TSP solution has a corresponding feasible SDP solution of cost $\TSPOPT(D)$).  The ratio $\f{\TSPOPT(D)}{\SDPOPT(D)}$ for any TSP cost matrix $D\in\mathcal{D}$ provides a lower bound on the integrality gap.

\subsection{SDP Relaxations}
 The QAP was introduced in Koopmans and Beckmann \cite{Koop57}.  Let matrices $A,B\in \mathbb{S}^{n\times n}$ respectively encode the pairwise distances between a set of $n$ locations and the pairwise flows between $n$ different facilities.  Let $C=(c_{ij})$ be a matrix of placement costs where $c_{ij}$ denotes the cost of placing facility $i$ at location $j$.   The QAP is to assign each facility to a distinct location so as to minimize total cost, where the cost depends quadratically on flows and distances and linearly on placement costs: 
 $$\min\{\tr\left((AXB+C)X^T\right): X\in \Pi_n\},$$ where $A, B\in \mathbb{S}^{n\times n}$ and $C\in \R^{n\times n}.$  The TSP for $n$ cities is obtained in the special case where $B=D, A=\f{1}{2}\mathcal{A}(\mathcal{C}_n)=\f{1}{2}C_1^{(n)},$ and $C=0$ (the all zeros matrix).  In this case, using the cyclic and linear properties of trace, the objective function becomes $$\tr\left(\f{1}{2}C_1^{(n)}XDX^T\right) = \f{1}{2}\langle X^TC_1^{(n)}X, D\rangle,$$ so that the permutation matrix $X$ can interpreted as finding the optimal tour and relabeling the vertices according to the order of that tour; $ X^TC_1^{(n)}X$ is then the adjacency matrix of the relabeled tour.

The SDP QAP relaxation of   Povh and Rendl \cite{Povh09}, when specialized to the TSP, is:
 \begin{equation}\label{eq:QAPToRed}
\begin{array}{l l l}
\min & \f{1}{2} \tr\left(\left(D\otimes C_1^{(n)} \right)Y\right) & \\
\text{subject to} & \tr((I_{n}\otimes E_{jj}^{(n)})Y)=1 & j=1, ..., n\\
 & \tr((E_{jj}^{(n)}\otimes I_{n})Y)=1 & j=1, ..., n\\
& \tr((I_{n}\otimes(J_{n}-I_{n})+(J_{n}- I_{n})\otimes I_{n})Y)=0 & \\
& \tr(J_{n^2}Y)=n^2 & \\
& Y \geq 0, Y\succeq 0, Y \in \mathbb{S}^{n^2\times n^2}.& 
\end{array} \end{equation}
That this is a valid relaxation can be seen by setting $Y=vec(X)vec(X)^T$ for any permutation matrix $X\in \Pi_n.$ Then letting $X_{:i}=X[[n], \{i\}]$ denote the $i$th column of $X$, 
$$vec(X)=\begin{pmatrix} X_{:1} \\ X_{:2} \\ \vdots \\ X_{:n}\end{pmatrix}$$  so that $Y$ has the block structure
$$Y=\begin{pmatrix} Y^{(11)} & Y^{(12)} & \cdots & Y^{(1n)} \\Y^{(21)} & Y^{(22)} & \cdots & Y^{(2n)} \\
\vdots & \vdots & \ddots & \vdots \\ Y^{(n1)} & Y^{(n2)} & \cdots & Y^{(nn)}  \end{pmatrix}$$ where $Y^{(ij)}=X_{:i}X_{:j}^T\in \R^{n\times n}.$
If $X$ is a permutation matrix, each $Y^{(ij)}=E_{st}^{(n)}$ for some $s, t$.  Specifically, $Y^{(ij)}=E_{st}^{(n)}$ for the $s, t$ such that $Xe_i=e_s$ and $Xe_j=e_t.$  That the constraints hold then readily follows: Each $Y^{(ii)}=E_{ss}^{(n)}$ for some $s$ (so that  $\tr((E_{ii}^{(n)}\otimes I_{n})Y)=1$) and because $X$ is a permutation matrix, $Y^{(ii)}\neq Y^{(kk)}$ for $i\neq k$ (so that $ \tr((I_{n}\otimes E_{jj}^{(n)})Y)=1$).  Similarly, each $Y^{(ii)}$ is diagonal while each $Y^{(ij)}$ with $i\neq j$ has zero diagonal (so $ \tr((I_{n}\otimes(J_{n}-I_{n})+(J_{n}- I_{n})\otimes I_{n})Y)=0$) and since each of the $n^2$ blocks $Y^{(ij)}$ consists of a single 1 and zeros elsewhere, the sum of all entries in $Y$ is $n^2$, i.e.\ $ \tr(J_{n^2}Y)=n^2.$  The factored form $Y=vec(X)vec(X)^T$ implies that $Y$ is a rank-1 positive semidefinite matrix and, since $Y$ is 0-1, $Y\geq 0.$  Finally,  $(Y^{(ij)})^T=(X_{:i}X_{:j}^T)^T=X_{:j}X_{:i}^T=Y^{(ji)}$ so that $Y$ is symmetric.   As we will show explicitly in Section \ref{Idea}, results from Gutekunst and Williamson \cite{Gut17} and  de Klerk et al.\ \cite{Klerk08} imply that SDP  (\ref{eq:QAPToRed})  has an unbounded integrality gap.

In de Klerk and Sotirov \cite{Klerk12b}, symmetry reduction is applied to  SDP  (\ref{eq:QAPToRed}) to obtain the following SDP relaxation of the TSP:
 \begin{equation}\label{eq:RedQAP}  \clonelabel{eq:RedQAP1up}{eq:RedQAP}
\hspace{-5mm} \begin{array}{l l l}
\min & \tr\left((D[\beta]\otimes \f{1}{2}C_1^{(n)}[\alpha]+Diag(\ol{c}))Y\right) & \\
\text{subject to} & \tr\left(\left(I_{n-1}\otimes E_{jj}^{(n-1)}\right)Y\right)=1 & j=1, ..., n-1\\
 & \tr\left(\left(E_{jj}^{(n-1)}\otimes I_{n-1}\right)Y\right)=1 & j=1, ..., n-1\\
& \tr\left(\left(I_{n-1}\otimes\left(J_{n-1}-I_{n-1}\right)+\left(J_{n-1}- I_{n-1}\right)\otimes I_{n-1}\right)Y\right)=0 & \\
& \tr((J_{n-1}\otimes J_{n-1})Y)=(n-1)^2 & \\
& Y \geq 0, Y\succeq 0, Y \in \mathbb{S}^{(n-1)^2\times (n-1)^2}.& 
\end{array} \end{equation}
where  $s, r\in [n],$ $\alpha=[n]\backslash r$ and $\beta=[n]\backslash s,$  and   $\ol{c}=vec(C_1[\alpha, \{r\}]D[\{s\}, \beta]).$ All that matters for the TSP is  the order in which the vertices are visited in the optimal tour; there are  $\f{(n-1)!}{2}$ distinct tours, but $n!$ permutation matrices.  One way to interpret the symmetry reduction intuitively is that, without loss of generality, one may assume an optimal solution $X\in \Pi_n$ is such that $X_{r, s}=1$ (i.e., that the $s$th vertex visited is vertex $r$): an optimal tour includes vertex $r$ and can be reindexed (without changing the cost of the tour) so that vertex $r$ is the $s$th vertex visited.  Making this assumption leaves the $n-1$ vertices $\alpha$ to be visited at the $n-1$ positions $\beta$, so one can effectively write a QAP for $X[\alpha, \beta]\in\Pi_{n-1}$ (the submatrix of $X$ for which entries are not fixed by $X_{r, s}=1$).  Following through this process obtains a QAP on $(n-1)$ vertices; appropriately adjusting the objective function and writing the SDP relaxation of the QAP on $(n-1)$ vertices yields the SDP relaxation  (\ref{eq:RedQAP}).  See de Klerk and Sotirov \cite{Klerk12b} for full details.

We will analyze the integrality gap of   SDP  (\ref{eq:RedQAP})  in Section \ref{Idea} (showing it is at least 2) and Section \ref{UBGap} (showing it is unbounded).  In both cases, we will find a set of instances on $n$ vertices and an associated feasible $Y\in \mathbb{S}^{n^2\times n^2}$ that together imply an unbounded integrality gap for  SDP  (\ref{eq:QAPToRed}).  We note that, up to dimension, the constraints of SDPs    (\ref{eq:QAPToRed}) and   (\ref{eq:RedQAP}) are exactly the same.  Any feasible $Y$ for an instance on $n$ vertices of SDP (\ref{eq:QAPToRed}) thus gives a feasible solution to SDP  (\ref{eq:RedQAP}), but to instances on $n+1$ vertices.  After finding an instance of $n$ vertices and feasible $Y$ for the SDP  (\ref{eq:QAPToRed}), our approach will be to add a single vertex and then use the same $Y$ to bound the integrality gap of SDP   (\ref{eq:RedQAP}) (accounting for the adjusted objective function).  It will thus be convenient to view  SDP  (\ref{eq:RedQAP}) as an SDP for $n+1$ vertex instances (with $n$ still even).  The SDP then becomes
 \begin{equation*}
\begin{array}{l l l}
\min & \tr\left((D[\beta]\otimes \f{1}{2}C_1^{(n+1)}[\alpha]+Diag(\ol{c}))Y\right) & \\
\text{subject to} & \tr((I_{n}\otimes E_{jj}^{(n)})Y)=1 & j=1, ..., n\\
 & \tr\left(\left(E_{jj}^{(n)}\otimes I_{n}\right)Y\right)=1 & j=1, ..., n\\
& \tr\left(\left(I_{n}\otimes\left(J_{n}-I_{n}\right)+\left(J_{n}- I_{n}\right)\otimes I_{n}\right)Y\right)=0 & \\
& \tr((J_{n}\otimes J_{n})Y)=n^2 & \\
& Y \geq 0, Y\succeq 0,  Y \in \mathbb{S}^{n^2\times n^2}.& 
\end{array} \end{equation*}
where  $s, r\in [n+1]$ and $\alpha=[n+1]\backslash r$ and $\beta=[n+1]\backslash s,$  and where  $\ol{c}=vec(C_1[\alpha, \{r\}]D[\{s\}, \beta]).$  We will also refer to this form of the SDP on $n+1$ vertices as SDP  (\ref{eq:RedQAP}).

\section{An Integrality Gap of At Least Two}\label{Idea}
We first show how results from Gutekunst and Williamson \cite{Gut17} and  de Klerk et al.\ \cite{Klerk08}  imply that the integrality gap of SDP (\ref{eq:QAPToRed}) is unbounded while the integrality gap of SDP  (\ref{eq:RedQAP}) is at least 2.   Theorem 3 of de Klerk et al.\ \cite{Klerk08} shows that the optimal value of SDP  (\ref{eq:QAPToRed}) coincides with the optimal value of an SDP relaxation of the TSP based on association schemes; Gutekunst and Williamson \cite{Gut17} give a family of instances that show this latter SDP has an unbounded integrality gap. By combining the same family of instances as Gutekunst and Williamson \cite{Gut17} and the relationship between the SDPs from Theorem 3 of  de Klerk et al.\ \cite{Klerk08}, we obtain that the integrality gap of SDP (\ref{eq:QAPToRed}) is unbounded while the integrality gap of SDP  (\ref{eq:RedQAP}) is at least 2.

\begin{thm}\label{thm:feas} Define
$$a_i = \f{2}{n-2} \left(\cos\left(\f{\pi i}{d}\right)+1\right), \hspace{5mm} i=1, ..., d,$$ and $$ b_i  = \begin{cases} \f{2}{n}\left(1-\cos\left(\f{\pi i}{d}\right)\right), & \text { if }i=1, ..., d-1 \\ \f{2}{n}, & \text{ if } i=d.\end{cases}$$  
Let $A=\sum_{i=1}^d a_iC_i$ and $B=\sum_{i=1}^d b_i C_i.$  Then $$Y=\f{1}{2n}\left(\left(I_2\otimes J_d - I_n\right) \otimes A + (J_2-I_2)\otimes J_d\otimes B + 2I_n\otimes I_n\right)$$ is feasible for SDP (\ref{eq:QAPToRed}).
\end{thm}

Note that $Y$ is an $n^2 \times n^2$ symmetric matrix that can be partitioned into blocks of size $n\times n$.  The $n$ blocks on the diagonal are scaled copies of the identity matrix.  The other blocks are all scaled copies of $A$ or $B$. For example, when $n=6$ we have $$Y=\f{1}{2n} \begin{pmatrix} 2I & A & A & B & B & B \\A& 2I & A & B & B & B \\A& A & 2I & B & B & B \\ B & B &B & 2I & A &A  \\ B & B &B & A & 2I &A  \\ B & B &B & A & A & 2I  \end{pmatrix}.$$

To Prove Theorem \ref{thm:feas}, we will make use of the following facts from  Gutekunst and Williamson \cite{Gut17}.  For completion, we sketch their proofs in the Appendix.  For $k=0,..., n-1,$ define $$a^{(k)} =\sum_{i=1}^d    \cos\left(\f{2\pi ik}{n}\right)  a_i, \hspace{5mm} b^{(k)} =\sum_{i=1}^d    \cos\left(\f{2\pi ik}{n}\right)  b_i.$$  Note that $$ \cos\left(\f{2\pi i(n-k)}{n}\right)= \cos\left(2\pi i - \f{2\pi ik}{n}\right)= \cos\left(\f{2\pi ik}{n}\right),$$ so that $a^{(k)}=a^{(n-k)}$ and $b^{(k)}=b^{(n-k)}.$
  
\begin{prop}\label{prop:recall}
\ \begin{enumerate}
\item $\sum_{i=1}^d a_i  = \sum_{i=1}^d b_i = 1.$  Equivalently, $a^{(0)}=b^{(0)}=1.$
\item $b^{(k)} =-\left(1-\f{2}{n}\right)a^{(k)} - \f{2}{n}.$
\item For $k=1, ..., d,$ $$a^{(k)} = \begin{cases} \f{d-2}{n-2}, & \text{ if } k=1 \\ -\f{2}{n-2}, & \text{ otherwise}. \end{cases}$$ 
\item $b_1 \leq \f{4\pi^2}{n^3}.$
\end{enumerate}
\end{prop}

To show that $Y$ is positive semidefinite, we will also use properties of circulant matrices.
\begin{lm}[Gray \cite{Gray06}]\label{lm:circ}
The circulant matrix $M= \left(m_{(t-s) \text{ mod } n}\right)_{s, t=1}^n$ has eigenvalues
 $$\lambda_t(M) = \begin{cases} \sum_{s=0}^{n-1} m_s e^{-\f{2\pi st  \sqrt{-1}}{n}}, & \text{ if } t=1, ..., n-1 \\  \sum_{s=0}^{n-1} m_s, & \text{ if }t=n.\end{cases}$$
 The eigenvector corresponding to eigenvalue $\lambda_t$ is $v_t = (1, w_t, w_t^2, ..., w_t^{n-1})$ for $t=0, 1, ..., n-1$ with $w_t=e^{-\f{2\pi t \sqrt{-1}}{n}}.$  
\end{lm}
\noindent To avoid ambiguity with index variables and imaginary numbers, we explicitly write $\sqrt{-1}$ whenever working with imaginary numbers and reserve $i$ and $j$ as index variables.

We first show that $Y$ satisfies each of the constraints of SDP  (\ref{eq:QAPToRed}).

\begin{cm}\label{cm1}
$\tr\left(\left(I_{n}\otimes E_{jj}\right)Y\right)=1$ and $ \tr\left(\left(E_{jj}\otimes I_{n}\right)Y\right)=1$ for $j=1, ..., n.$
\end{cm}
\begin{proof}
Each of the $n^2$ diagonal entries of $Y$ is $\f{1}{n}.$  Both $I_{n}\otimes E_{jj}$ and $E_{jj}\otimes I_{n}$ are diagonal matrices with exactly $n$ nonzero entries, all of which are equal to 1. \hfill
\end{proof}

\begin{cm}
$\tr((I_{n}\otimes(J_{n}-I_{n})+(J_{n}- I_{n})\otimes I_{n})Y)=0.$
\end{cm}
\begin{proof}The $n\times n$ blocks of $Y$ have  sparsity patterns that imply this constraint:  $I$ is a diagonal matrix, while $A$ and $B$ have zero diagonal (there is no coefficient of $C_0$ in the sums defining $A$ and $B$).   \hfill
\end{proof}

\begin{cm}\label{cm:easysum}
$ \tr(J_{n^2}Y)=n^2.$
\end{cm}
\begin{proof}
To show this constraint holds, we note that $Y$ is expressed in terms of $n^2$ blocks, each of size $n\times n$ and each of which is either  $\f{1}{2n}A, \f{1}{2n}B,$ or $\f{1}{n}I.$  In the first row of $A$, we have that $A_{1, i}=a_i=a_{n-i}$ for $i=1, ..., d-1,$ while $A_{1, d}=2a_d$.  Since $A$ is circulant, each of the $n$ rows of $A$ then sums to $2\sum_{i=1}^d a_i.$  Using the first result of Proposition \ref{prop:recall}, the entries in $A$ sum to $2n$ so that $\tr(J_n\f{1}{2n}A)=1.$  Analogously, $\tr(J_n \f{1}{2n}B)=\tr(J_n\f{1}{n}I_n)=1$. That is, each of the $n^2$ blocks defining $Y$ sums to 1 so that, when we sum all the entries in $Y$, $$\tr(J_{n^2}Y)=n^2.$$
\hfill
\end{proof}

\begin{cm}\label{cm:nn1}
$Y\geq 0.$
\end{cm}
\begin{proof}
The penultimate constraint follows because $a_i, b_i\geq 0.$ \hfill
\end{proof}

To show feasibility, we thus must finally show
\begin{cm}\label{cm2}
$Y\succeq 0.$
\end{cm}

\begin{proof} 
From Lemma \ref{lm:circ}, we have that the eigenvectors of a circulant matrix with first row $(m_0, m_1, ..., m_{n-1})$ are of the form $v_j = (1, w_j, w_j^2, ..., w_j^{n-1})$ for $j=0, 1, ..., n-1$ with $w_j=e^{-\f{2\pi j \sqrt{-1}}{n}}.$  The  eigenvalue corresponding to $v_j$ is $$\lambda_j=m_0+m_{1}w_j+m_{2}w_j^2+\cdots m_{n-1}w_j^{n-1}.$$ Hence, $v_j$ is a simultaneous eigenvector of $A, B,$ and $I_n.$  Let $\lambda_j^A$ and $\lambda_j^B$ respectively indicate the eigenvalues of $A$ and $B$ corresponding to $v_j$.  Note that
\begin{align*}
w_j^i + w_j^{n-i} & = e^{-\f{2\pi j i \sqrt{-1}}{n}}+e^{-\f{2\pi j (n-i)\sqrt{-1}}{n}}\\
& =\cos\left(-\f{2\pi i j}{n}\right)+\sqrt{-1} \sin\left(-\f{2\pi i j}{n}\right) + \cos\left(-\f{2\pi (n-i) j}{n}\right)+ \sqrt{-1} \sin\left(-\f{2\pi (n-i) j}{n}\right)\\
&= 2 \cos\left(\f{2\pi i j}{n}\right).
\end{align*}
 Then since $A_{1, i}=A_{1, n-i}=a_i$ for $i=1, ..., d-1$ and $A_{1, d}=2a_d$,  $$\lambda_j^A =\left( \sum_{i=1}^{d-1}a_i(w_j^i+w_j^{n-i})\right)+2a_dw_j^d=2\sum_{i=1}^{d} a_i \cos\left(\f{2\pi i j }{n}\right) = 2a^{(j)}.$$ 
 Similarly, $\lambda_j^B=2b^{(j)}.$

Recall that $$Y=\f{1}{2n}\left(\left(\left(I_2\otimes J_d\right) - I_n\right) \otimes A + (J_2-I_2)\otimes J_d\otimes B + 2I_n\otimes I_n\right).$$  By finding a shared set of eigenvectors of $(I_2\otimes J_d)-I_n,$ $(J_2-I_2)\otimes J_d$ and $2I_n,$ we can use properties of the Kronecker product to explicitly compute the eigenvalues of $Y$ as a function of the $a^{(j)}$ and $b^{(j)}$; the remaining results from Proposition \ref{prop:recall} will suffice to show that they are all nonnegative. We will use the following as our shared set of eigenvectors\footnote{
To find this shared set of eigenvectors, note that $J_m=e^{(m)}(e^{(m)})^T$ is a rank-1 matrix and that $$J_me^{(m)}=e^{(m)}(e^{(m)})^Te^{(m)}=me^{(m)}.$$ The only nonzero eigenvector of $J_m$ is thus $e^{(m)}$ with corresponding eigenvalue $m$. All other eigenvectors have corresponding eigenvalue zero, and a convenient basis for them is $e^{(m)}_1-e^{(m)}_i$ for $i=2, ..., m.$  Then $$J_m\left(e^{(m)}_1-e^{(m)}_i\right)=e^{(m)}-e^{(m)}=0\left(e^{(m)}_1-e^{(m)}_i\right).$$  The vectors $e^{(m)}, e^{(m)}_1-e^{(m)}_2, ..., e^{(m)}_1-e^{(m)}_m$ are linearly independent and so form an eigenbasis for $J_m$.  To extend these to find eigenvectors of $(I_2\otimes J_d)-I_n,$ $(J_2-I_2)\otimes J_d$ and $2I_n,$ we use 1) the spectral properties of Kronecker products noted in the introduction, and 2) the fact that if $v$ is an eigenvector of a matrix $M$ with corresponding eigenvalue $\lambda$ then $v$ is also an eigenvector of $M-I$ with corresponding eigenvalue $\lambda-1:$ $$(M-I)v=\lambda v-v=(\lambda-1)v.$$
 }.  We first have $u_1=e^{(n)}$ and $u_2 = (e_1^{(2)}-e_2^{(2)})\otimes e^{(d)}.$
The remaining  $u_3, ..., u_n$ are the $n-2$ vectors of the form $e^{(2)}\otimes (e_1^{(d)}-e_i^{(d)})$ and $(e^{(2)}_1-e^{(2)}_2)\otimes (e_1^{(d)}-e_i^{(d)})$  for $i=2, ..., d$ (in any order).    
Denote by $\mu_j^A$ and $\mu_j^B$ the respective eigenvalues of $(I_2\otimes J_d)-I_n$ and $(J_2-I_2)\otimes J_d$ associated 
with $u_j$.  Then $$\mu_1^A = d-1, \hspace{5mm} \mu_2^A= d-1,  \hspace{5mm} \mu_j^A= -1 \text{ otherwise}$$ and $$\mu_1^B = d, \hspace{5mm} \mu_2^B= -d,  \hspace{5mm} \mu_j^B= 0 \text{ otherwise}.$$

Now note that 
\begin{align*}&
\left(\left(\left(I_2\otimes J_d\right)- I_n\right) \otimes A + (J_2-I_2)\otimes J_d\otimes B + 2I_n\otimes I_n\right)(u_i\otimes v_j)=(\mu_i^A\lambda_j^A+\mu_i^B\lambda_j^B+2)(u_i\otimes v_j),
\end{align*}
so that the eigenvalues of $2nY$ must be  the values of $(\mu_i^A\lambda_j^A+\mu_i^B\lambda_j^B+2)$ over $i=1, ..., n$ and $j=0, ..., n-1$.  That is, 
 $$2(\mu_i^Aa^{(j)}+\mu_i^Bb^{(j)}+1), \hspace{5mm} i=1, ..., n; j=0, ..., n-1.$$  To show $Y\succeq 0,$ it suffices to show that these are all nonnegative. For $j=0,$ we have that $a^{(0)}=b^{(0)}=1$ and thus that 
$$\mu_i^Aa^{(0)}+\mu_i^Bb^{(0)}+1 = \begin{cases} d-1+d+1=2d\geq 0, & i=1 \\ d-1-d+1 = 0 \geq 0, & i=2 \\ -1 + 0 + 1 = 0 \geq 0, & i\geq 3. \end{cases}$$

Otherwise, for $j\neq 0,$ we have 
\begin{align*}
\mu_i^Aa^{(j)}+\mu_i^Bb^{(j)}+1&=  \begin{cases} (d-1)a^{(j)}+db^{(j)} + 1, & i=1 \\(d-1)a^{(j)}-db^{(j)}+1, & i=2 \\-a^{(j)}+1, & i\geq3. \end{cases}
\intertext{Using  $b^{(k)} =-\left(1-\f{2}{n}\right)a^{(k)} - \f{2}{n}$ from  Proposition \ref{prop:recall}:}
&= \begin{cases} 0, & i=1 \\ (n-2)a^{(j)}+2 & i=2 \\-a^{(j)}+1, & i\geq 3. \end{cases}
\end{align*}
 By the final case of Proposition \ref{prop:recall}, for $j=1, ..., d,$ $$a^{(j)} = \begin{cases} \f{d-2}{n-2}, & \text{ if } j=1 \\ -\f{2}{n-2}, & \text{ otherwise}. \end{cases}$$ Hence, the eigenvalues are all nonnegative and $Y\succeq 0.$
\hfill
\end{proof}

\begin{proof}[Proof (of Theorem \ref{thm:feas})]
Claims \ref{cm1} to \ref{cm2} imply $Y$ is feasible for SDP (\ref{eq:QAPToRed}). \hfill
\end{proof}
  
\begin{cor}\label{cor:UB1}
The integrality gap of SDP  (\ref{eq:QAPToRed}) is unbounded.
\end{cor}

 To show that the integrality gap is unbounded, we consider the cost matrix $$D=\begin{pmatrix} 0 & 1 \\ 1 & 0 \end{pmatrix} \otimes J_d$$ used in Gutekunst and Williamson \cite{Gut17}.  This cost matrix is that of a \emph{cut semi-metric}: there are two equally sized groups of vertices $\{1, ..., d\}$ and $\{d+1, ..., n\};$ the cost of traveling between two vertices in the same group is zero, while the cost of traveling between two vertices in different groups (i.e. crossing the cut defined by $\{1, ..., d\}$) is 1.  
 
 \begin{proof}
 
 The integrality gap of  SDP  (\ref{eq:QAPToRed}) is at least $$\f{\TSPOPT(D)}{\SDPOPT(D)}.$$  Note that $\TSPOPT(D)\geq 2$, as a minimum cost Hamiltonian cycle must cross the cut defined by $\{1, ..., d\}$ twice; the tour $1, 2, ..., d, d+1, ..., n, 1$ realizes this cost so that $\TSPOPT(D)= 2.$  We then bound $\SDPOPT(D)$ using $Y$ as a feasible solution to SDP  (\ref{eq:QAPToRed}).  Note that, when computing the cost, we evaluate $ \tr\left(\left(D\otimes C_1 \right)Y\right).$  The $n^2\times n^2$ matrix $D\otimes C_1$ consists of $n\times n$ blocks, either of which is an $n\times n$ block of zeros (exactly where $Y$ has an $\f{1}{2n}A$ block or a $\f{1}{n}I$ block) or a $C_1$ (exactly in the $2d^2$ places where $Y$ has a $\f{1}{2n}B$ block).  Hence:
 \begin{align*} 
\SDPOPT(D) & \leq   \f{1}{2} \tr\left((D\otimes C_1 )Y\right) \\
& = \f{1}{2} 2d^2 \f{1}{2n} \tr(C_1 B) \\
& = \f{d^2}{2n} 2nb_1 \\
&= d^2 b_1 \\
&\leq 4\pi^2 \f{d^2}{n^3},
\end{align*}
using the final result of Proposition \ref{prop:recall}.   Thus $\SDPOPT(D)\leq c\f{1}{n}$ for some constant $c$.  Hence the integrality gap is at least
$$\f{\TSPOPT(D)}{\SDPOPT(D)}\geq \f{2}{c\f{1}{n}} = \f{2}{c}n,$$ which grows without bound.
\hfill
 \end{proof}
 
This instance and feasible solution $Y$ do not, however, show that the integrality gap of  SDP  (\ref{eq:RedQAP}) is unbounded.  Instead, they imply the following.
 
\begin{cor}\label{cor:IG2}
The integrality gap of SDP  (\ref{eq:RedQAP}) is at least 2.
\end{cor}

We consider an instance of  SDP  (\ref{eq:RedQAP}) on $n+1$ vertices.  This change of bookkeeping implies that the feasible solution $Y$ from Theorem \ref{thm:feas} is feasible for  SDP  (\ref{eq:RedQAP1up}).

\begin{proof}
We consider an instance on $n+1$ vertices with two groups of vertices, $\{1, ..., d, d+1\}$ and $\{d+2, ..., n+1\}.$  As before, we define the cost of traveling between vertices in the same group to be zero and the cost of traveling between vertices in distinct groups to be 1.  By taking $r=s=1$ we have that $$D[\beta] = \begin{pmatrix} 0 & 1 \\ 1 & 0 \end{pmatrix} \otimes J_d$$ and $C_1^{(n+1)}[\alpha] \leq C_1^{(n)}$ entrywise.  As in  Corollary \ref{cor:UB1}, the integrality gap is at least
 $$\f{\TSPOPT(D)}{\SDPOPT(D)},$$
 and again $\TSPOPT(D)=2.$  To upper bound the denominator, we note that feasibility of $Y$ implies
\begin{align*}
\SDPOPT(D) &\leq \tr\left(\left(D[\beta]\otimes \f{1}{2}C_1^{(n+1)}[\alpha]+Diag(\ol{c})\right)Y\right) \\ 
&= \tr\left(\left(D[\beta]\otimes \f{1}{2}C_1^{(n+1)}[\alpha]\right)Y\right)+\tr(Diag(\ol{c})Y).
\end{align*}
We can bound the first term by Corollary \ref{cor:UB1}, since
$$\tr\left(\left(D[\beta]\otimes \f{1}{2}C_1^{(n+1)}[\alpha]\right)Y\right) \leq \tr\left(\left(\left( \begin{pmatrix} 0 & 1 \\ 1 & 0 \end{pmatrix} \otimes J_d\right)\otimes \f{1}{2}C_1^{(n)}\right)Y\right) \leq c\f{1}{n}.$$
We can  compute the second term.  Note that $C_1[\al, \{r\}]=e^{(n)}_1+e^{(n)}_n$ and $D[\{s\}, \beta]=(e_1^{(n)}+e_2^{(n)}+...+e_d^{(n)})^T.$  Thus  $C_1[\al, \{r\}] D[\{s\}, \beta]$ is an $n\times n$ matrix with exactly $n$ ones and all other entries zero.  Hence $$Diag(\ol{c})=Diag(vec(C_1[\al, \{r\}] D[\{s\}, \beta]))$$ is a diagonal matrix with exactly $n$ ones on the diagonal.  Since each diagonal entry of $Y$ is $\f{1}{n},$ we have $$\tr(Diag(\ol{c})Y)=1.$$  Putting everything together, we get that
$$\SDPOPT(D)\leq 1+\f{c}{n},$$ so that the integrality gap is at least
 $$\f{\TSPOPT(D)}{\SDPOPT(D)}\geq \f{2}{1+\f{c}{n}}=\f{2n}{n+c}$$ for some constant c, which gets arbitrarily close to $2$ as $n$ grows.
\hfill
\end{proof}
Note also that the solution $Y$ is not necessarily optimal for  SDP (\ref{eq:RedQAP1up}), and hence this family of instances may in fact imply an integrality gap larger then 2.  Numerical experiments on this family indicate that the optimal solutions to  SDP (\ref{eq:RedQAP1up}) have value strictly less than $1$ as $n$ grows sufficiently large, but are far less structured then $Y.$  To show that the integrality gap of SDP (\ref{eq:RedQAP1up}) is unbounded, we instead modify the family of instances considered. 

We will specifically look for instances $D$ where $\TSPOPT(D)$ grows arbitrarily large while $\SDPOPT(D) \leq a+\f{b}{n},$ for constants $a$ and $b$; as in the previous proof, we will bound $\tr(Diag(\ol{c})Y)$ by an absolute constant ($a$) and show that $ \tr(D[\beta]\otimes \f{1}{2}C_1^{(n+1)}[\alpha])Y$ decays with $\f{1}{n}$ (with a constant $b$ that does not depend on $n$).    We will also want to find feasible (but not necessarily optimal) solutions that retain the structure of $Y$: a matrix with a simple block structure that respects that of the cost matrix; can be decomposed into terms, each of which is the Kronecker product of a matrix constructed using $J$ and $I$ and a circulant matrix; and that we will thus be able to explicitly write down their eigenvalues.

\section{The Unbounded Integrality Gap} \label{UBGap}
In this section we prove our main theorem:
\begin{thm}\label{thm:main}
Let $z\in\N$.  Then the integrality gap of SDP  (\ref{eq:RedQAP1up}) is at least $z.$
\end{thm}
\noindent An immediate corollary is:
\begin{cor}
The integrality gap of SDP  (\ref{eq:RedQAP1up}) is unbounded.
\end{cor}

As before, we first start by finding feasible solutions to SDP (\ref{eq:QAPToRed}).  We then modify these solutions so that they are feasible to SDP (\ref{eq:RedQAP1up}).

\subsection{Feasible Solutions to SDP (\ref{eq:QAPToRed})}
To generalize the above example, we consider an instance with $g$ equally sized groups of $\f{n}{g}$ vertices. If $u, v$ are two vertices in the same group, then the cost of traveling between $u$ and $v$ is zero; otherwise the cost is 1.   Labeling the vertices so that the $i$th group consists of vertices $\{(i-1)\f{n}{g}+1, ..., i\f{n}{g}\},$ the cost matrix is $$D=(J_g-I_g)\otimes J_{n/g}.$$ Note that the instances in Section \ref{Idea} are the special case when $g=2.$  Note also that these instances are metric and can be viewed as Euclidean TSP in $\R^{g-1};$ we refer to this family of instances as \emph{simplicial TSP instances}: In a regular $g-1$ simplex, there are $g$ extreme points, each pair of which is a distance 1 apart.  One way to interpret an instance with $g$ groups is as embedded into a regular $g-1$ simplex in $\R^{g-1}$ where each group of $\f{n}{g}$ vertices is placed at an extreme point of the simplex.

To prove Theorem \ref{thm:main}, we will take $g=2z.$  To simplify the our proofs, we thus assume that $g$ is even throughout.

We use solutions of the form \begin{equation}\label{eq:Y}Y=\f{1}{2n}\left[(J_g-I_g)\otimes J_{n/g}\otimes B + I_g\otimes J_{n/g} \otimes A+I_g\otimes I_{n/g}\otimes (2I_n-A)\right],\end{equation} where $$A=\sum_{i=1}^d a_i C_i, \hspace{5mm} B=\sum_{i=1}^d b_i C_i$$ are symmetric circulant matrices defined in terms of parameters $a_1, ..., a_d$ and $b_1, ..., b_d.$  We set
$$a_i = \begin{cases} \f{1}{n-g} \left[2 + \f{4}{g} \sum_{j=1}^{g-1} (g-j) \cos\left(\f{\pi i j}{d}\right)\right], & i<d \\ 
\f{1}{n-g} \left[1 + \f{2}{g} \sum_{j=1}^{g-1} (g-j) \cos\left(\f{\pi i j}{d}\right)\right], & i=d. \end{cases}$$
We also set\footnote{These values come from assuming $$\left(\f{n}{g}-1\right)a_i + \f{g-1}{g}n b_i = \begin{cases} 2, & i<d \\ 1, & i=d.\end{cases}$$ The intuition for choosing these values of $b_i$ is to impose an analogue of the degree constraints from Gutekunst and Williamson \cite{Gut17}.
}
$$b_i = \begin{cases}  \f{2g-(n-g)a_i}{n(g-1)}, & i<d \\  \f{g-(n-g)a_i}{n(g-1)}, & i=d. \end{cases}$$
We will often take sums of the $a_i$ or $b_i$.  It will be helpful to note that
$$a_d = \f{1}{n-g} \left[2 + \f{4}{g} \sum_{j=1}^{g-1} (g-j) \cos\left(\f{\pi d j}{d}\right)\right]-\f{1}{n-g} \left[1 + \f{2}{g} \sum_{j=1}^{g-1} (g-j) \cos\left(\pi j\right)\right]$$ and $$b_d = \f{2g-(n-g)a_d}{n(g-1)}-\f{g}{n(g-1)}.$$

Figure \ref{fig:ai} provides intuition for how the $a_i$ depend on $g$.  The $a_i$ can be viewed as uniform samples from a sum of cosines that places a larger weight on smaller values of $i$.  As $g$ increases, the proportion of the $a_i$ that are close to zero grows.  As in Section \ref{Idea}, the only parameter that $\SDPOPT$ depends on will be $b_1$, and large $a_1$ implies small $b_1.$

\begin{figure}[t]
\centering
\includegraphics[scale=0.6]{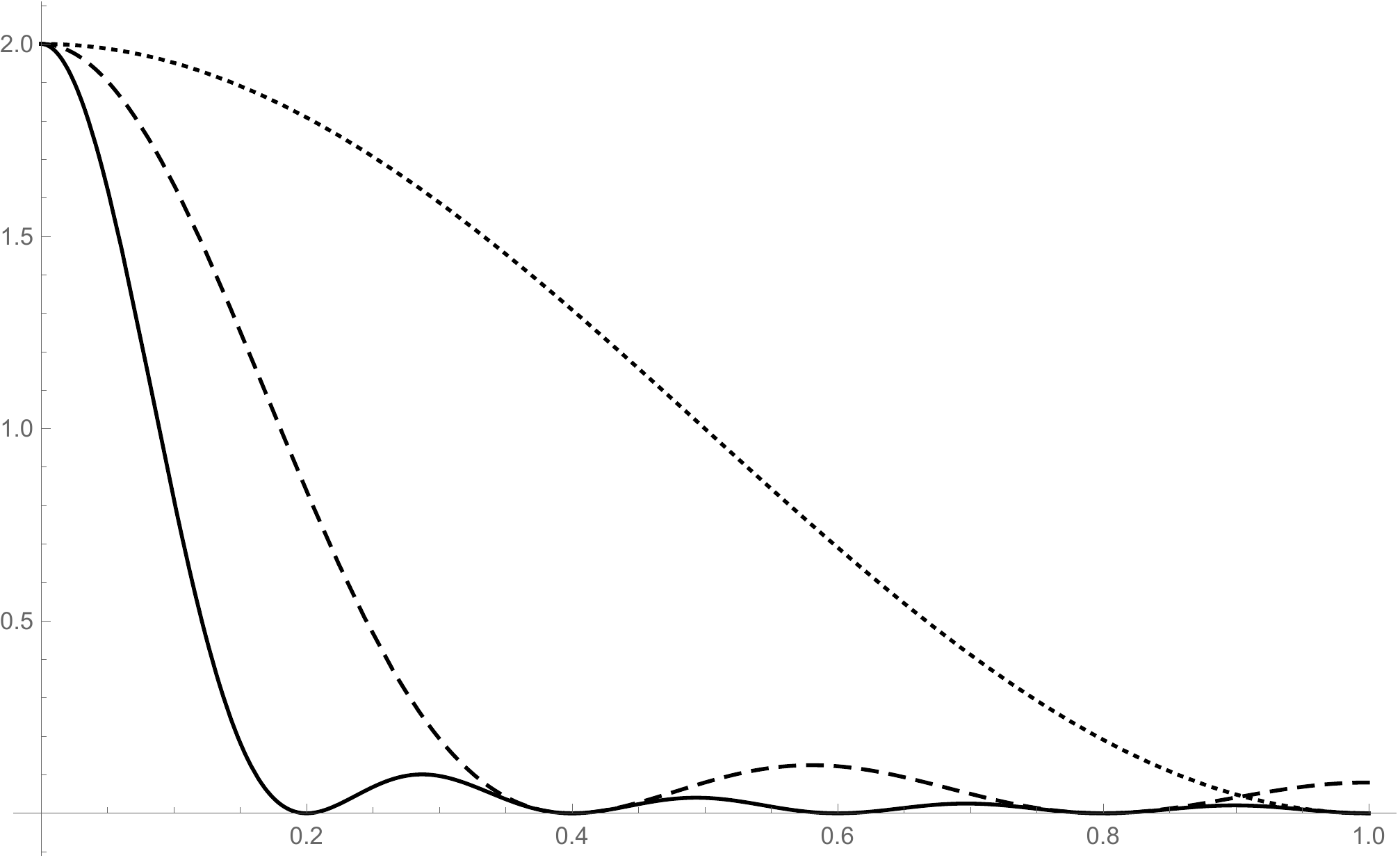} \caption{$\f{n-g}{g}a_i$ for $g=2, 5, $ and $10$. For each curve (and any value of $n$), the values $a_1, ..., a_{d-1}$ are taken by sampling the curve at $x=\f{1}{d}, \f{2}{d}, ..., \f{d-1}{d}$; the value of $a_d$ is half the value at $x=1$.  The dotted curve shows $g=2,$ the dashed curve shows $g=5$, and the remaining curve shows $g=10$. }\label{fig:ai}
\end{figure}

Note that $Y$ is a large block matrix that respects symmetry of our cost matrix $D$ in the exact same way as in Section \ref{Idea}: each diagonal block is $\f{1}{n}I_n$; everywhere else that $D$ has a 0, $Y$ places a block $\f{1}{2n}A$; everywhere $D$ has a 1, $Y$ has a block $\f{1}{2n}B.$  In the proofs below, it will help to refer to multiple types of blocks of $Y$.   $Y$ can  be partitioned into larger blocks of size $\f{n^2}{g}\times \f{n^2}{g}$, each of which  is either  $J_{n/g}\otimes \f{1}{2n} B$ or $\f{1}{2n}\left((J_{n/g}\otimes A)+I_{n/g}\otimes (2I_n-A)\right)$; we will refer to these blocks as \emph{major blocks}.  The former are off-diagonal, so we will refer to them as \emph{major off-diagonal blocks} while the latter are on the diagonal of $Y$, so we will refer to them as \emph{major diagonal blocks}.   Each of these major blocks consists of $(n/g)^2$ smaller, $n\times n$ blocks, each of which is a $\f{1}{2n}A, \f{1}{2n}B,$ or  $\f{1}{n}I_n.$  We will refer to each as a \emph{minor block}.  We refer to each of the $n$ blocks of $\f{1}{n}I_n$ as a \emph{minor diagonal block}, and the remaining $n\times n$ blocks (each of which is a single $n\times n$ block equal to $\f{1}{2n}A$ or $\f{1}{2n}B$) as a \emph{minor off-diagonal block}. 
\begin{exam} Suppose $g=3$ and $n=12.$  Pictorially,  the minor blocks are those blocks proportional to $I_n, A,$ and $B$; the major blocks are those delineated below that each consist of 16 minor blocks.
$$ Y=\f{1}{2n}\left(
\begin{array}{cccc | cccc | cccc}
 2 I_n & A & A & A & B & B & B & B & B & B & B & B \\
 A & 2 I_n & A & A & B & B & B & B & B & B & B & B \\
 A & A & 2 I_n & A & B & B & B & B & B & B & B & B \\
 A & A & A & 2 I_n & B & B & B & B & B & B & B & B \\ \hline
 B & B & B & B & 2 I_n & A & A & A & B & B & B & B \\
 B & B & B & B & A & 2 I_n & A & A & B & B & B & B \\
 B & B & B & B & A & A & 2 I_n & A & B & B & B & B \\ 
 B & B & B & B & A & A & A & 2 I_n & B & B & B & B \\ \hline
 B & B & B & B & B & B & B & B & 2 I_n & A & A & A \\
 B & B & B & B & B & B & B & B & A & 2 I_n & A & A \\
 B & B & B & B & B & B & B & B & A & A & 2 I_n & A \\
 B & B & B & B & B & B & B & B & A & A & A & 2 I_n \\
\end{array}
\right)$$
\end{exam}
We now show that this solution meets each SDP constraint.  

\begin{prop}\label{prop:f1}
$Y,$ as given by Equation (\ref{eq:Y}), is feasible for  SDP (\ref{eq:QAPToRed}).
\end{prop}

\begin{cm}\label{cm:41}
For each $j=1, ..., n,$ we have 
$$ \tr((I_{n}\otimes E_{jj})Y)=1 \hspace{5mm} \text{and} \hspace{5mm}  \tr((E_{jj}\otimes I_{n})Y)=1.$$
\end{cm}
\begin{proof}
Note that these constraints only impact the diagonal entries of $Y$, each of which is equal to $\f{1}{n}.$ The constraint $ \tr((I_{n}\otimes E_{jj})Y)=1$ expands as $$Y_{j,j} + Y_{n+j, n+j} +  Y_{2n+j, 2n+j}+\cdots +  Y_{(n-1)n+j, (n-1)n+j} = 1.$$  The constraint $  \tr((E_{jj}\otimes I_{n})Y)=1$ expands as $$Y_{(j-1)n+1, (j-1)n+1}+Y_{(j-1)n+2, (j-1)n+2}+\cdots+Y_{(j-1)n+n, (j-1)n+n}=1.$$  Both summands consist of $n$ terms, each of which is equal to $\f{1}{n},$ so both hold immediately.
\hfill
\hfill \end{proof}

\begin{cm}
$\tr((I_{n}\otimes(J_{n}-I_{n})+(J_{n}- I_{n})\otimes I_{n})Y)=0.$
\end{cm}
\begin{proof}
This constraint holds because of $Y$'s sparsity pattern:  First note that $$\tr((I_{n}\otimes(J_{n}-I_{n}))Y)=0,$$ as each $n\times n$ minor diagonal block of $Y$ is $\f{1}{n}I_n$, which is diagonal.  Second $$\tr(((J_{n}- I_{n})\otimes I_{n})Y)=0,$$  as every minor diagonal block is either $\f{1}{2n}A$ or $\f{1}{2n}B;$ the matrices $A$ and $B$ are a linear combination of $C_1, ..., C_d$ all of which have every diagonal entry zero.
\hfill
\hfill \end{proof}

\begin{cm}\label{prop:sum}
$\tr(J_{n^2}Y)=(n)^2.$
\end{cm}

This proof involves some involved bookkeeping and uses a handful of lemmas.  We use $\bi_{\{\circ\}}$ to denote the indicator function that is 1 if event $\circ$ happens and zero otherwise.

\begin{lm}\label{lem:trig}
Let $n$ be even and $0<k<n$ be an integer.  Then $$\sum_{j=1}^d \cos\left(\f{\pi jk}{d}\right) = \f{-1+(-1)^k}{2}.$$  
\end{lm}

\noindent This identity is a consequence of Lagrange's trigonometric identity; see the Appendix for a more detailed proof.

\begin{lm}\label{lem:sums} \label{lem:altsums}
Let $g$ be even.  Then 
$$ \sum_{j=1}^{g-1} (g-j) \bi_{\{j \text{ odd}\}}=\f{g^2}{4} \hspace{5mm} \text{ and } \hspace{5mm}  \sum_{j=1}^{g-1} (g-j) (-1)^j = -\f{g}{2}.$$
 \end{lm}
 
  \begin{proof}
The first claim of this  lemma readily follows from the fact that the sum of the first $m$ positive odd integers is $m^2$.
 \begin{align*}
  \sum_{j=1}^{g-1} (g-j) \bi_{\{j \text{ odd}\}} &=  (g-1) + (g-3) + ... + 1 = \left(\f{g}{2}\right)^2,
  \end{align*}
  where we note that we added $\f{g}{2}$ odd numbers.  The second claim follows since:
  \begin{align*}
 \sum_{j=1}^{g-1} (g-j) (-1)^j &=  \left[-(g-1)+(g-2)\right]+\left[-(g-3)+(g-4)\right]+\cdots+ \left[-3+ 2\right]-1,
\\
  &= -1\f{g-2}{2} - 1 \\&= -\f{g}{2}.
 \end{align*}
 \hfill \end{proof}

\begin{lm}\label{lem:aisum}
$\sum_{i=1}^d a_i = 1.$
\end{lm}
\begin{proof}
This lemma follows by direct computation using the preceding identities.
\begin{align*}
\sum_{i=1}^d a_i & =\left(\sum_{i=1}^d \f{1}{n-g} \left[2+\f{4}{g}\sum_{j=1}^{g-1}(g-j)\cos\left(\f{\pi i j}{d}\right) \right] \right) - \f{1}{n-g} \left[1+\f{2}{g} \sum_{j=1}^{g-1} (g-j) \cos\left( \pi j\right) \right] \\
&= \f{1}{n-g} \left( 2d + \f{4}{g}\left[\sum_{i=1}^d \sum_{j=1}^{g-1} (g-j)\cos\left(\f{\pi i j}{d}\right)\right] - 1 - \f{2}{g}\sum_{j=1}^{g-1} (g-j) (-1)^j \right) \\
&= \f{1}{n-g} \left( 2d + \f{4}{g}\left[ \sum_{j=1}^{g-1} (g-j) \sum_{i=1}^d \cos\left(\f{\pi i j}{d}\right)\right] - 1 - \f{2}{g}\sum_{j=1}^{g-1} (g-j) (-1)^j \right).
\intertext{By Lemma \ref{lem:trig}:}
&= \f{1}{n-g} \left( 2d + \f{4}{g}\left[ \sum_{j=1}^{g-1} (g-j) \f{(-1) + (-1)^j}{2} \right] - 1 - \f{2}{g}\sum_{j=1}^{g-1} (g-j) (-1)^j \right) \\ 
&= \f{1}{n-g} \left( 2d - \f{4}{g}\left[ \sum_{j=1}^{g-1} (g-j) \bi_{\{j \text{ odd}\}} \right] - 1 - \f{2}{g}\sum_{j=1}^{g-1} (g-j) (-1)^j \right). 
\intertext{By Lemma \ref{lem:sums}, and using that $g$ is even:}
&= \f{1}{n-g}\left(
 2d - \f{4}{g}\left[ \f{g^2}{4} \right] - 1 + \f{2}{g}\f{g}{2} \right) \\
  &= \f{1}{n-g} (2d-g) \\
  &= 1,
\end{align*}
since $n=2d.$
\hfill \end{proof}

\begin{lm}\label{lm:bsum2}
$\sum_{i=1}^d b_i = 1.$
\end{lm}
\begin{proof}  
This lemma readily follows from the definition of the $b_i$ in terms of the $a_i$. 
\begin{align*}
\sum_{i=1}^d b_i &= \left( \sum_{i=1}^d \f{2g-(n-g)a_i}{n(g-1)}\right) - \f{g}{n(g-1)}\\
&= \f{2gd}{n(g-1)} - \f{n-g}{n(g-1)}\sum_{i=1}^d a_i  - \f{g}{n(g-1)}.
\intertext{By Lemma \ref{lem:aisum}:}
&= \f{2gd}{n(g-1)} - \f{n-g}{n(g-1)} - \f{g}{n(g-1)}\\
&= \f{1}{n(g-1)} \left(ng-n+g-g\right)\\
&= \f{1}{n(g-1)} \left(n(g-1)\right)\\
&= 1.
\end{align*}
\hfill \end{proof}

\begin{proof}[Proof (of Claim \ref{prop:sum})]
To show that $\tr(J_{n^2} Y)=n^2,$ we want to sum the entries of $Y$.  We mirror the proof of Claim \ref{cm:easysum} and first compute the sum of the entries in each minor block, which is either a $\f{1}{n}I_n, \f{1}{2n}A,$ or $\f{1}{2n}B$.  As in Claim \ref{cm:easysum}, Lemma \ref{lem:aisum} implies that $\tr(J_n \f{1}{2n}A)=\f{1}{2n}2n\sum_{i=1}^d a_i = \f{1}{2n}2n=1,$ and analogously Lemma \ref{lm:bsum2} implies that $\tr(J_n \f{1}{2n} B)=1.$  Moreover, $\tr(J_n \f{1}{n}I_n)=1$.  Hence, each of the $n^2$  minor blocks of $Y$ sums to 1, so that total sum of entries in $Y$ is
$$\tr(J_{n^2}Y) = n^2.$$ \hfill \end{proof}

\begin{cm}\label{prop:nonneg}
$Y\geq 0.$
\end{cm}

To show that $Y\geq 0,$ we show that the $a_i$ and $b_i$ are nonnegative.  We will use the following trigonometric identity.

\begin{lm}\label{cm:nonnegid}
$$\left(2\cos(\theta)-2\right)\sum_{j=1}^{g-1}(g-j)\cos(j\theta)=  \cos(g\theta)-g\cos(\theta)+(g-1).$$
\end{lm}

\begin{proof}
\begin{align*}
& \left(2\cos(\theta)-2\right)\sum_{j=1}^{g-1}(g-j)\cos(j\theta) \\
&\hspace{10mm} = 2\sum_{j=1}^{g-1}(g-j)\cos(j\theta)\cos(\theta)-2\sum_{j=1}^{g-1}(g-j)\cos(j\theta).\intertext{Applying the product-to-sum identity for cosine:}
&\hspace{10mm} =  \sum_{j=1}^{g-1}(g-j)\cos((j+1)\theta) + \sum_{j=1}^{g-1}(g-j)\cos((j-1)\theta) -2 \sum_{j=1}^{g-1}(g-j)\cos(j\theta). \intertext{Reindexing to combine terms:}
&\hspace{10mm} =  \sum_{j=2}^{g}(g-j+1)\cos(j\theta) + \sum_{j=0}^{g-2}(g-j-1)\cos(j\theta) - 2\sum_{j=1}^{g-1}(g-j)\cos(j\theta)\\
&\hspace{10mm} = \sum_{j=1}^{g-1}\left[(g-j+1)+(g-j-1)-2(g-j)\right]\cos(j\theta) +\cos(g\theta)-g\cos(\theta)+(g-1)\cos(0)-0\\
&\hspace{10mm} =   \cos(g\theta)-g\cos(\theta)+(g-1).
\end{align*}
\hfill \end{proof}

\begin{proof}[Proof (of Claim \ref{prop:nonneg})]
We first show that the $a_i$ are nonnegative. Recall that 
$$a_i  \propto 2 + \f{4}{g} \sum_{j=1}^{g-1} (g-j) \cos\left(\f{\pi i j}{d}\right)$$ (where the constant of proportionality is different for $a_1, ..., a_{d-1}$ and for $a_d$, but in both cases is positive).
To show that the $a_i$ are nonnegative, we thus want to show that, for $i=1, ..., d$,
$$ \f{4}{g} \sum_{j=1}^{g-1} (g-j) \cos\left(\f{\pi i j}{d}\right) \geq -2,$$ or equivalently
$$\sum_{j=1}^{g-1} (g-j) \cos\left(\f{\pi i j}{d}\right) \geq -\f{g}{2}.$$
We appeal to Lemma \ref{cm:nonnegid} with $\theta = \f{\pi i}{d}.$  For $i=1, ..., d,$ $\cos(\theta)\neq 1,$ so we have that:
\begin{align*}
\sum_{j=1}^{g-1} (g-j) \cos\left(\f{\pi i j}{d}\right) &=  \f{\cos(g\theta)-g\cos(\theta)+g-1}{2\cos(\theta)-2} \\
&= \f{g(1-\cos(\theta))}{2(\cos(\theta)-1)}+\f{\cos(g\theta)-1}{2\cos(\theta)-2}\\
&= -\f{g}{2} + \f{1-\cos(g\theta)}{2-2\cos(\theta)} \\
& \geq -\f{g}{2},
\end{align*}
since $1-\cos(g\theta)\geq 0$ and $2-2\cos(\theta)\geq 0$.

We now need only show that the $b_i\geq 0$. 
Recall that 
$$b_i = \begin{cases}  \f{2g-(n-g)a_i}{n(g-1)}, & i<d \\  \f{g-(n-g)a_i}{n(g-1)}, & i=d. \end{cases}$$  For $i=1, ..., d-1$ it suffices to show that $2g\geq (n-g)a_i.$  In these cases, we have
\begin{align*}
(n-g)a_i &= 2 +  \f{4}{g} \sum_{j=1}^{g-1} (g-j) \cos\left(\f{\pi i j}{d}\right).
\intertext{Using $\cos(\theta)\leq 1$:}
&\leq  2 +  \f{4}{g} \sum_{j=1}^{g-1} (g-j) \\
&= 2 + \f{4}{g} (1+2+...+(g-1)) \\
&= 2 + \f{4}{g} \f{(g-1)g}{2} \\
&= 2 + 2(g-1) \\
&= 2g,
\end{align*}
as desired.

For $i=d$ the situation is analogous.  We want $(n-g)a_d \leq g$ which follows by the exact computations as above.
\begin{align*}
(n-g)a_d &= \f{1}{2}\left[2 +  \f{4}{g} \sum_{j=1}^{g-1} (g-j) \cos\left(\f{\pi i j}{d}\right) \right] \\
&\leq \f{1}{2} 2g\\
&= g.
\end{align*}
\hfill
\end{proof}

\begin{prop}\label{prop:Ypsd}
$Y\succeq 0.$
\end{prop}

As before, define $a^{(k)} = \sum_{i=1}^d a_i \cos\left(\f{2\pi i k}{n}\right)$ and $b^{(k)} = \sum_{i=1}^d b_i \cos\left(\f{2\pi i k}{n}\right).$   Recall, as in  Claim \ref{cm2}, that the eigenvectors of a general circulant matrix are of the form $v_j = (1, w_j, w_j^2, ..., w_j^{n-1})$ for $j=0, 1, ..., n-1$

\begin{cm}\label{cm:evals}
A and B are simultaneously diagonalizable.  The eigenvalues of $A$ are $$\lambda_k(A)=2a^{(k)}$$ for $k=0, ..., n-1.$ The eigenvalues of $B$ are $$\lambda_k(B)=2b^{(k)}$$ for $k=0, ..., n-1,$ where $\lambda_k(A)$ and $\lambda_k(B)$ correspond to the same eigenvector $v_k.$
\end{cm}

\begin{proof}This is exactly as in Claim \ref{cm2}, since $A$ and $B$ are constructed using the same basis of symmetric circulant matrices.
\hfill \end{proof}

\begin{cm} The distinct eigenvalues of $2nY$ are
$$\begin{cases}
2(g-1)\f{n}{g}b^{(k)}+2\f{n}{g}a^{(k)}+(2-2a^{(k)})\\
-2\f{n}{g}b^{(k)}+2\f{n}{g}a^{(k)}+(2-2a^{(k)})\\
2-2a^{(k)},\\
\end{cases}$$ over $k=0, ..., n-1.$
\end{cm}

\begin{proof}
Note that
$$2nY=(J_g-I_g)\otimes J_{n/g}\otimes B + I_g\otimes J_{n/g} \otimes A+I_g\otimes I_{n/g}\otimes (2I_n-A).$$  Claim \ref{cm:evals} gives a set of simultaneous eigenvectors/eigenvalues for $B$ and $A$ (and thus also $2I_n-A$) which we denote by $v_k$, for $k=1, ..., n$. We can similarly obtain a simultaneous set of eigenvectors/eigenvalues of $(J_g-I_g)\otimes J_{n/g}, I_g\otimes J_{n/g}$ and $I_g\otimes I_{n/g},$ so that we will again  use properties of the Kronecker product to explicitly compute the eigenvalues of $Y$ as a function of the $a^{(k)}$ and $b^{(k)}$. Note $(J_g-I_g)\otimes J_{n/g} $ has three distinct eigenvalues:  $J_g-I_g$ has two distinct eigenvalues ($g-1$ with associated eigenvector $e^{(g)}$ and $-1$ with associated eigenvectors $e_1^{(g)}-e_i^{(g)}$, for $i=2, ..., g$) and $J_{n/g}$ has two distinct eigenvalues ($n/g$ with associated eigenvector $e^{(n/g)}$ and $0$ with associated eigenvectors $e^{(n/g)}_1-e^{(n/g)}_i$ for $i=2, ..., \f{n}{g}$). Hence spectral products of Kronecker products imply that the distinct eigenvalues of $(J_g-I_g)\otimes J_{n/g}$ are $$\mu_i^B:=\begin{cases} (g-1)\times\f{n}{g}, & i=1  \text{ (using } e^{(g)} \otimes e^{(n/g)}) 
\\ -1\times\f{n}{g}, &i=2  \text{ (using } \left(e_1^{(g)}-e_i^{(g)} \right) \otimes e^{(n/g)}) 
\\ (g-1)\times0=-1\times 0, & i=3 \text{ (using } e^{(g)} \otimes \left( e^{(n/g)}_1-e^{(n/g)}_i \right) \text{ or } \left(  e_1^{(g)}-e_i^{(g)}\right) \otimes \left( e^{(n/g)}_1-e^{(n/g)}_i\right))
.\end{cases}$$  In exactly the same way, the distinct eigenvalues of $I_g\otimes J_{n/g}$ are $$\mu_i^A:=\begin{cases} 1\times\f{n}{g}, & i=1  \text{ (using } e^{(g)} \otimes e^{(n/g)}) 
\\ 1\times\f{n}{g}, &i=2  \text{ (using } \left( e_1^{(g)}-e_i^{(g)} \right) \otimes e^{(n/g)}) 
\\ 1\times0, & i=3 \text{ (using } e^{(g)} \otimes \left( e^{(n/g)}_1-e^{(n/g)}_i \right)\text{ or }   \left(e_1^{(g)}-e_i^{(g)}\right) \otimes \left( e^{(n/g)}_1-e^{(n/g)}_i\right))
.\end{cases}$$  

For $1\leq i\leq 3,$ let $u_i$ be a shared eigenvector of $(J_g-I_g)\otimes J_{n/g}$ and  $I_g\otimes J_{n/g} \otimes A$ with respective associated eigenvalues $\mu_i^B$ and $\mu_i^A$.  Then:
$$\left((J_g-I_g)\otimes J_{n/g}\otimes B + I_g\otimes J_{n/g} \otimes A+I_g\otimes I_{n/g}\otimes (2I_n-A)\right)(u_i\otimes v_k) =( \mu_i^B\lambda_k^B + \mu_i^A\lambda_k^A+(2-\lambda_k^A))(u_i\otimes v_k).$$ 
Plugging in for the three cases of $\mu_i^A$ and $\mu_i^B$, we get that the distinct eigenvalues of $2nY$ are
\begin{equation} \label{eq:evals}
\begin{cases}
2(g-1)\f{n}{g}b^{(k)}+2\f{n}{g}a^{(k)}+(2-2a^{(k)})\\
-2\f{n}{g}b^{(k)}+2\f{n}{g}a^{(k)}+(2-2a^{(k)})\\
2-2a^{(k)},\\
\end{cases} \end{equation}
over $k=0, ...., n-1$ as claimed \hfill
\end{proof}

\begin{cm}\label{cm:akbkrel}  For $k=1, ..., n-1$,
$$b^{(k)}= -\f{g}{n(g-1)}- \f{n-g}{n(g-1)}a^{(k)}.$$
\end{cm}
\begin{proof}
We have that
\begin{align*}
b^{(k)} &= \sum_{i=1}^d b_i \cos\left(\f{2\pi i k}{n}\right) \\
&= \left[\sum_{i=1}^d \f{2g-(n-g)a_i}{n(g-1)}\cos\left(\f{2\pi i k}{n}\right)\right] - \f{g}{n(g-1)} \cos(\pi k). \intertext{Using Lemma \ref{lem:trig} and the definition of $a^{(k)}$:}
&= \f{2g}{n(g-1)}\left[\sum_{i=1}^d \cos\left(\f{2\pi i k}{n}\right) \right] - \f{n-g}{n(g-1)}\left[\sum_{i=1}^d a_i \cos\left(\f{2\pi i k}{n}\right) \right] - \f{g}{n(g-1)} (-1)^k\\
&= \f{g}{n(g-1)}\left((-1)+(-1)^k\right) - \f{n-g}{n(g-1)}a^{(k)}- \f{g}{n(g-1)} (-1)^k\\
&= -\f{g}{n(g-1)}- \f{n-g}{n(g-1)}a^{(k)}.\\
\end{align*}
\hfill \end{proof}

Plugging in for the $b^{(k)}$ we can simplify the eigenvalues of $2nY$.  
\begin{cm}\label{cm:evalsa}
The eigenvalues of $2nY$ are 
$$\begin{cases}
0\\
 2\left(\f{g}{g-1}\right) + 2\left(\f{n-g}{g-1}\right)a^{(k)}\\
2-2a^{(k)},\\
\end{cases}$$
over $k=1, ..., n-1$ and
$$\begin{cases}
2n\\
0,\\
\end{cases}$$
corresponding to $k=0$.
\end{cm}

\begin{proof}
The $k=0$ follows by simplifying Equation (\ref{eq:evals}) using $a^{(0)}=b^{(0)}=1$ from Lemmas \ref{lem:aisum} and \ref{lm:bsum2}.  Otherwise, notice that
\begin{align*}
2(g-1)\f{n}{g}b^{(k)}+2\f{n}{g}a^{(k)}+(2-2a^{(k)}) 
&= -2(g-1)\f{n}{g}\left(\f{g}{n(g-1)}+ \f{n-g}{n(g-1)}a^{(k)}\right)+2\f{n}{g}a^{(k)}+(2-2a^{(k)}) \\
&= 2\left[-1-\f{n-g}{g}a^{(k)}+\f{n}{g}a^{(k)} + 1 - \f{g}{g} a^{(k)}\right]\\
&= 0.
\end{align*}
Similarly 

\begin{align*}
-2\f{n}{g}b^{(k)}+2\f{n}{g}a^{(k)}+(2-2a^{(k)})
&= 2\f{n}{g}\left(\f{g}{n(g-1)}+ \f{n-g}{n(g-1)}a^{(k)}\right)+2\f{n}{g}a^{(k)}+(2-2a^{(k)}) \\
&= 2\left(\f{1}{g-1}+\f{n-g}{g(g-1)}a^{(k)}\right)+2\f{n}{g}a^{(k)}+(2-2a^{(k)}) \\
&= 2\left(1+\f{1}{g-1}\right) + 2\left(\f{n-g}{g(g-1)}+\f{n(g-1)}{g(g-1)}-\f{g(g-1)}{g(g-1)}\right)a^{(k)} \\
&= 2\left(1+\f{1}{g-1}\right) + 2\left(\f{ng-g^2}{g(g-1)}\right)a^{(k)} \\
&= 2\left(\f{g}{g-1}\right) + 2\left(\f{n-g}{g-1}\right)a^{(k)}.\\
\end{align*}
\hfill
\end{proof}

To complete the proof of Proposition \ref{prop:Ypsd}, we will show that the eigenvalues for the cases $k=1, ..., n-1$ are nonnegative.  We will use two lemmas.

\begin{lm}\label{cm:akLB}
$$\sum_{i=1}^d \cos\left(\f{\pi i j}{d}\right)\cos\left(\f{\pi i k}{d}\right) \geq -\bi_{j-k \text{ odd}}.$$
\end{lm}

\begin{proof}
By the product-to-sum identity,
\begin{align*}
\sum_{i=1}^d \cos\left(\f{\pi i j}{d}\right)\cos\left(\f{\pi i k}{d}\right)
&= \f{1}{2} \left( \sum_{i=1}^d \cos\left(\f{\pi i (j+k)}{d}\right) + \cos\left(\f{\pi i (j-k)}{d}\right)\right). \intertext{Applying Lemma \ref{lem:trig} and considering separately the cases where we cannot apply it (those that devolve down to summing $\cos(\theta i)$ over $i$ when $\theta$ is an integer multiple of $2\pi$):}
&= \begin{cases} \f{1}{4} \left[-1+(-1)^{j+k} + -1 + (-1)^{j-k}\right], & j-k, j+k\notin \{0, n\} \\
\f{1}{4} \left[-1+(-1)^{j+k} + 2d \right], & j-k\in \{0, n\}, j+k\notin \{0, n\}\\
\f{1}{4} \left[2d -1+(-1)^{j-k} \right], & j-k\notin \{0, n\}, j+k\in \{0, n\}\\
\f{1}{4} \left[2d +2d \right], & j-k, j+k \in \{0, n\}.
\end{cases} 
\intertext{Noting that $(-1)^{j+k}=(-1)^{j-k}$:}
&\geq \f{1}{2} \left(-1 + (-1)^{j-k}\right) \\
&= -\bi_{j-k \text{ odd}}.
\end{align*}

Note that $j-k, j+k\in \{0, n\}$ requires $j-k=0$ and $j+k=n$, i.e. $j=k=d.$  Since $j$ ranges from $1$ to $g-1,$ our final case is irrelevant if each group contains at least 2 vertices.

\hfill \end{proof}

\begin{lm}\label{cm:bjsum} For $g$ even,
$$\sum_{j=1}^{g-1}(g-j)\bi_{j-k \text{ odd}}  = 
 \f{1}{4}\left( g(g-1)+g(-1)^k\right).$$
\end{lm}

\begin{proof}
\begin{align*}
\sum_{j=1}^{g-1}(g-j)\bi_{j-k \text{ odd}} &= -\f{1}{2} \sum_{j=1}^{g-1} \left(-1+(-1)^{j-k}\right)(g-j)\\
 &= \f{1}{2} \sum_{j=1}^{g-1}(g-j) -\f{1}{2}\sum_{j=1}^{g-1} (-1)^{j-k}(g-j)\\
 &= \f{1}{2} (1+2+...+(g-1)) - \f{1}{2}(-1)^{-k} \sum_{j=1}^{g-1} (-1)^j (g-j). 
 \intertext{Using Lemma \ref{lem:altsums}:}
 &= \f{1}{2} \f{g(g-1)}{2} + \f{1}{2}(-1)^k \f{g}{2}\\
  &=  \f{1}{4}\left( g(g-1)+g(-1)^k\right).
\end{align*}

\hfill \end{proof}

\begin{proof}[Proof (of Proposition \ref{prop:Ypsd})]
To complete the proof of Proposition \ref{prop:Ypsd}, we need only show that the eigenvalues listed in Claim \ref{cm:evalsa} are nonnegative.  We thus need to show that $$ 2\left(\f{g}{g-1}\right) + 2\left(\f{n-g}{g-1}\right)a^{(k)}\geq 0 \hspace{5mm} \text{ and } \hspace{5mm} 2-2a^{(k)}\geq 0$$ for $k=1, ..., n-1.$  The latter is a direct consequence of Claim \ref{prop:nonneg}, since
$a_i\geq 0$ implies 
$$a^{(k)}=\sum_{i=1}^d a_i \cos\left(\f{2\pi i k}{n}\right)  \leq \sum_{i=1}^d a_i =1.$$

Hence we need only show that  $ 2\left(\f{g}{g-1}\right) + 2\left(\f{n-g}{g-1}\right)a^{(k)}\geq 0.$  Equivalently, we need to show that $$a^{(k)} \geq-\f{g}{n-g}.$$ This result holds since: 

\begin{align*}
\hspace{-4mm} a^{(k)} &= \sum_{i=1}^d a_i \cos\left(\f{2\pi i k}{n}\right) \\
&= \f{1}{n-g} \left[\left[\sum_{i=1}^d\left(2+ \f{4}{g} \sum_{j=1}^{g-1} (g-j) \cos\left(\f{\pi i j}{d}\right)\right)\cos\left(\f{2\pi i k}{n}\right) \right] - \left[1+ \f{2}{g} \sum_{j=1}^{g-1} (g-j) \cos\left(\f{\pi d j}{d}\right)\right]\cos\left(\f{2\pi d k}{n}\right)\right] 
\intertext{By Lemma \ref{lem:trig}:}
&= \f{1}{n-g} \left[\left[(-1)+(-1)^k+\f{4}{g}\sum_{j=1}^{g-1}(g-j)\sum_{i=1}^d \cos\left(\f{\pi i j}{d}\right)\cos\left(\f{2\pi i k}{n}\right) \right]- \left[1+ \f{2}{g} \sum_{j=1}^{g-1} (g-j)(-1)^j\right](-1)^k\right].
\intertext{By Lemma \ref{cm:akLB}:}
&\geq \f{1}{n-g} \left[\left[(-1)+(-1)^k-\f{4}{g}\sum_{j=1}^{g-1}(g-j)\bi_{j-k \text{ odd}}\right]- \left[1+ \f{2}{g} \sum_{j=1}^{g-1} (g-j)(-1)^j\right](-1)^k\right]. 
\intertext{By Lemmas \ref{lem:altsums} and \ref{cm:bjsum}:}
&= 
\f{1}{n-g} \left[\left[(-1)+(-1)^k-\f{1}{g} \left( g(g-1)+ g(-1)^k\right)\right]- \left[1- \f{2}{g} \f{g}{2}\right](-1)^k\right]\\
&= \f{1}{n-g} \left[(-1)+(-1)^k- (g-1)-(-1)^k \right]\\
&=-\f{g}{n-g}.
\end{align*}
\hfill \end{proof}

 \begin{proof}[Proof (of Proposition \ref{prop:f1})]
 Feasibility of $Y$ follows directly from Claims \ref{prop:f1} to \ref{prop:sum}, Claim \ref{prop:nonneg}, and Proposition \ref{prop:Ypsd}. \hfill
 \end{proof}

We can also compute the objective function value of $Y$.

\begin{thm}\label{thm:SDP1gap}
For $Y$ as above, there exists a constant $\tilde{c}_g$ (depending on $g$ but not $n$) such that $$\f{1}{2} \tr((D\otimes C_1 )Y)\leq \f{\tilde{c}_g}{n}. $$
\end{thm}

\begin{proof}
Recalling that $D=(J_g-I_g)\otimes J_{n/g},$ we see that $D\otimes C_1$ has block of zeros in each of the $g$ major $\f{n^2}{g} \times \f{n^2}{g}$ diagonal blocks of $Y$.  Hence the only places where $D\otimes C_1$  places a nonzero entry are exactly those where $Y$ has a $B$ block; on each such block, $D\otimes C_1$ has a block $C_1$.  There are $g(g-1)$ blocks of $B$ matrices, each containing $\f{n^2}{g^2}$ copies of $B$.  Accounting for the fact that $Y$ is scaled by $\f{1}{2n}$, the value of the objective function is thus

\begin{align*}
\f{1}{2} \tr((D\otimes C_1 )Y)
&= \f{1}{2}g(g-1)\f{n^2}{g^2} \f{1}{2n}\tr(C_1  B).
\intertext{Since $\tr(C_1 B)=2nb_1$: }
&=\f{1}{2}g(g-1)\f{n^2}{g^2}  b_1 \\
&= \f{1}{2} \f{g-1}{g} n^2 b_1.
\end{align*}
Recall that
 $$\cos\left(x\right) \geq 1-\f{1}{2}x^2.$$
 Hence
 \begin{align*}
 b_1 &=\f{2g-(n-g)a_1}{n(g-1)} \\
 &= \f{2g- \left[2 + \f{4}{g} \sum_{j=1}^{g-1} (g-j) \cos\left(\f{\pi  j}{d}\right)\right]}{n(g-1)} \\
  &\leq \f{2g- \left[2 + \f{4}{g} \sum_{j=1}^{g-1} (g-j) \left(1-\f{1}{2}\f{\pi^2j^2}{d^2} \right)\right]}{n(g-1)} \\
    &= \f{2(g-1)-\f{4}{g}\sum_{j=1}^{g-1}(g-j) +\f{2}{g}\f{\pi^2}{d^2}\sum_{j=1}^{g-1}(g-j)j^2}{n(g-1)}. 
    \intertext{Define $c_g=\f{2}{g}\pi^2\sum_{j=1}^{g-1}(g-j)j^2,$ a constant depending on $g$ but not $n$.}
    &= \f{2(g-1)-\f{4}{g}\f{(g-1)g}{2} + \f{c_g}{d^2}}{n(g-1)} \\
        &= \f{c_g}{d^2n(g-1).} 
        \intertext{Setting $\hat{c}_g=\f{4}{g-1}c_g,$ a constant depending on $g$ but not $n$:}
                &= \f{\hat{c}_g}{n^3}. \\
 \end{align*} 
 Putting everything together, 
 $$\f{1}{2} \tr((D\otimes C_1 )Y)= \f{1}{2}\f{g-1}{g} n^2 b_1 \leq \f{1}{2} \f{g-1}{g} n^2 \f{\hat{c}_g }{n^3},$$
from which the result follows. \hfill
\end{proof}

We can now prove our main theorem, which we restate below.
\begin{thm*}[Theorem \ref{thm:main}]
Let $z\in\N$.  Then the integrality gap of SDP  (\ref{eq:RedQAP1up}) is at least $z.$
\end{thm*}
\begin{proof}
We again consider the SDP  (\ref{eq:RedQAP}) corresponding to an instance on $n+1$ vertices. Let $s=r=1$ and consider an instance of the TSP on $n+1$ vertices with $g=2z$ groups of vertices.  Specifically, let groups $2, ...., g$ be equally sized, each of size $\f{n}{g}\in\N$, and let group 1 have one extra vertex, so that group one is of size $\f{n}{g}+1$.   Note also that $$\TSPOPT=g=2z$$ since each group of vertices must be visited at least once. 
Set $$Y=\f{1}{2n}\left[(J_g-I_g)\otimes J_{n/g}\otimes B + I_g\otimes J_{n/g} \otimes A+I_g\otimes I_{n/g}\otimes (2I_n-A)\right],$$ which is feasible for the SDP by our earlier computations.  Then the integrality gap is bounded below by
$$\f{\TSPOPT}{\SDPOPT} \geq \f{2z}{\tr((D[\beta]\otimes \f{1}{2}C_1^{(n+1)}[\alpha]+Diag(\ol{c}))Y) }.$$
To bound the right-hand side, we note that linearity of the trace operator implies
\begin{equation}\label{eq:MT1}\tr(D[\beta]\otimes \f{1}{2}C_1^{(n+1)}[\alpha]+Diag(\ol{c}))Y = \tr((D[\beta]\otimes \f{1}{2}C_1^{(n+1)}[\alpha] )Y)+\tr(Diag(\ol{c})Y).\end{equation}
We upper bound each term. 
First note that $$D[\beta]=D[\{2, ..., n+1\}]=(J_g-I_g)\otimes J_{(n/g)}.$$  Similarly, $$C_1^{(n+1)}[\alpha]=C_1^{(n)}-e_1e_n^T-e_ne_1^T \leq C_1^{(n)},$$ where $\leq$ is taken entry-wise.  By non-negativity, 
\begin{equation}\label{eq:MT2} \tr((D[\beta]\otimes \f{1}{2}C_1^{(n+1)}[\alpha] )Y) \leq \f{1}{2}\tr((((J_g-I_g)\otimes J_{(n/g)})\otimes(C_1^{(n)}))Y)\leq \f{\tilde{c}_g}{n},\end{equation} by Theorem \ref{thm:SDP1gap}.  As in  Theorem \ref{thm:SDP1gap} $\tilde{c}_g$ remains independent of $n$.

Second, consider $\tr(Diag(\ol{c})Y).$  We compute that
$$\ol{c}=vec(C_1^{(n)}[\al, \{1\}] D[\{1\}, \beta]) = vec((e^{(n)}_1+e^{(n)}_n) D[\{1\}, \beta]) \leq vec(e^{(n)}_1+e^{(n)}_n) (e^{(n)})^T.$$  We note that $(e^{(n)}_1+e^{(n)}_n) (e^{(n)})^T$ is an $n\times n$ matrix with $2n$ ones and the rest of the entries zero.  The $vec$ operator stacks the columns of this matrix, creating a vector in $\R^{n^2}$ with $2n$ ones and the remaining entries zero.  Finally, $Diag(\ol{c})$ creates a diagonal matrix with $2n$ ones on the diagonal and the remaining entries zero.  Since all diagonal entries of $Y$ are equal to $\f{1}{n},$ we have that
\begin{equation}\label{eq:MT3}
\tr(Diag(\ol{c})Y) \leq 2n*\f{1}{n}=2.
\end{equation}
Plugging Equations (\ref{eq:MT2}) and (\ref{eq:MT3}) into Equation (\ref{eq:MT1}) we obtain 
$$\tr\left(\left(D[\beta]\otimes \f{1}{2}C_1^{(n+1)}[\alpha]+Diag(\ol{c})\right)Y \right)\leq \f{\tilde{c}_g}{n}+2.$$  Hence the integrality gap is at least:
\begin{align*}
\f{\TSPOPT}{\SDPOPT} &\geq \f{2z}{\tr((D[\beta]\otimes \f{1}{2}C_1^{(n+1)}[\alpha]+Diag(\ol{c}))Y )}\\
&\geq \f{2z}{2+\f{\tilde{c}_g}{n}} \\
&= \f{2zn}{2n+\tilde{c}_g} \\
&\rar z,
\end{align*}
as $n\rar \infty.$
\hfill
\end{proof}

Gutekunst and Williamson \cite{Gut17}  show that the SDPs of Cvetkovi{\'c} et al.\ \cite{Cvet99} and de Klerk et al.\ \cite{Klerk08} have a counterintuitive non-monotonicity property: adding vertices (in a way that retains costs being metric) can arbitrarily decease the cost of some solutions to the corresponding SDPs.  This property contrasts with both TSP and subtour LP solutions:  monotonicity of the TSP can be seen through \emph{shortcutting} (see, e.g., Section 2.4 of Williamson and Shmoys \cite{DDBook}), while Shmoys and Williamson \cite{Shm90} show that the subtour LP is monotonic.   Corollary \ref{cor:IG2} shows that in the $g=2$ case the cost of the SDP, $$\SDPOPT(D)\leq 1+\f{c}{n},$$ decays arbitrarily close to 1 as the number of vertices in each group grew. Any $g=2$ instance with cost strictly greater than one thus shows that  SDP (\ref{eq:RedQAP}) is non-monotonic.  Moreover, such an instance implies the non-montonicity property in $\R^1$: the SDP can find a smaller optimal value by only adding more points to visit on the real line.

\begin{cor}
The SDP (\ref{eq:RedQAP}) is non-monotonic.
\end{cor}

\begin{proof}[Proof (sketch)]
It  suffices to show a single two group instance with cost strictly greater than 1.  Consider such an instance on $n+1=3$ vertices, where the first group has two vertices and the second has one.  Explicitly writing down the constraints shows that any feasible solution to the SDP has cost 2.
\hfill
\end{proof}

Anstreicher \cite{Ans00} gives  another SDP relaxation of the QAP, and our simplicial instances also show that its integrality gap is unbounded.  In the case where $C=0$ and $e$ is an egeinvector of either data matrix in the QAP objective function, their SDP is equivalent to the projected eigenvalue bound of Hadley,  Rendl, and Wolkowicz \cite{Had92}.  Since $C_1^{(n)}e^{(n)}=2e^{(n)},$ it is equivalent to the projected eigenvalue bound when specialized to the TSP.  In this case, the SDP is in terms of an $n^2\times n^2$ matrix which we give block structure
$$Y=\begin{pmatrix} Y^{(11)} & Y^{(12)} & \cdots & Y^{(1n)} \\Y^{(21)} & Y^{(22)} & \cdots & Y^{(2n)} \\
\vdots & \vdots & \ddots & \vdots \\ Y^{(n1)} & Y^{(n2)} & \cdots & Y^{(nn)}  \end{pmatrix}$$
with $Y^{(ij)}\in \R^{n\times n}.$  The SDP is:
 \begin{equation}\label{eq:AQAP}
 \begin{array}{l l l}
\min & \f{1}{2} \tr\left(\left(D\otimes C_1^{(n)} \right)Y\right) & \\
\text{subject to} &\sum_{i=1}^n Y^{(ii)} &= I_n\\
 &\left( \tr(Y^{(ij)})\right)_{i, j=1}^n& = I_n\\
& \tr\left(Y F^TF\right)&=2n \\
& Y - \f{1}{n^2}J_{n^2} &\succeq 0 \\
& Y \geq 0, Y \in \mathbb{S}^{n^2\times n^2},& 
\end{array} \end{equation}
where $$F=\begin{pmatrix}  (e^{(n)})^T \otimes I_n \\ I_n \otimes (e^{(n)})^T \end{pmatrix}.$$

Let $Y'=vec(X)vec(X)^T$ for any $X\in \Pi_n$; that this is a valid relaxation can be seen by showing that $$Y'-\f{1}{n^2}\left(Y'J_{n^2}+J_{n^2}Y'\right)+\f{2}{n^2}J_{n^2}$$ is feasible and has the same objective value as $Y'$.  See Theorem 3.6 of Anstreicher \cite{Ans00} for more details.  

\begin{cor}
 SDP (\ref{eq:AQAP}) has an unbounded integrality gap.
\end{cor}

\begin{proof}
We show that $Y$, as defined in Theorem \ref{thm:feas}, remains feasible.  The objective function remains unchanged from SDP  (\ref{eq:QAPToRed}), so the analysis in Corollary \ref{cor:UB1} then implies that SDP (\ref{eq:AQAP}) has an unbounded integrality gap.

By definition of $Y$, $Y^{(ii)}=\f{1}{n}I_n$ for $i=1, ..., n$ so that $\sum_{i=1}^n Y^{(ii)} = I_n.$  Moreover $\tr(Y^{(ii)})=\f{1}{n}\tr(I_n)=1$ for $i=1, ..., n,$ while $A$ and $B$ have zero diagonal so that the trace of any minor off-diagonal block is zero.  Hence $\left( \tr(Y^{(ij)})\right)_{i, j=1}^n = I_n.$  

Next, note that $F^TF=J_n\otimes I_n + I_n \otimes J_n.$  Thus $F^TF$ has a $J_n+I_n$ on each minor diagonal block and an $I_n$ on each minor off-diagonal block.  Since $A$ and $B$ have zero diagonal, $\tr(AI_n)=\tr(BI_n)=0$ and the minor off-diagonal blocks make no contribution to $\tr\left(Y F^TF\right).$  Hence 
$$\tr\left(Y F^TF\right) = n \tr\left(\f{1}{n}I_n (J_n+I_n)\right) = \tr\left(I_n(J_n+I_n)\right)=2n.$$

By Claim \ref{cm:nn1} and the definition of $Y$, $Y$ is nonnegative and symmetric.  Hence it remains to show that $Y - \f{1}{n^2} J_{n^2} \succeq 0.$  We note that $e^{(n^2)}$ is an eigenvector of $Y.$ In the notation of Claim \ref{cm2}, it is the eigenvector when $j=0$ and $i=1.$  In  Claim \ref{cm2}, we showed that the corresponding eigenvalue of $nY$ was $2d=n,$ so that the corresponding eigenvalue of $Y$ is $1.$  Then
$$\left(Y - \f{1}{n^2}J_{n^2}\right)e^{(n^2)}=Ye^{(n^2)}-\f{1}{n^2}e^{(n^2)}\left(e^{(n^2)}\right)^Te^{(n^2)} = e^{(n^2)}-\f{1}{n^2}e^{(n^2)}n^2 = 0e^{(n^2)}.$$ Any other eigenvector $v$ of $Y$ is orthogonal to $e^{(n^2)}.$  Letting $\lambda$ denote the corresponding eigenvalue, 
$$\left(Y - \f{1}{n^2}J_{n^2}\right)v=Yv-\f{1}{n^2}e^{(n^2)}\left(e^{(n^2)}\right)^Tv =\lambda v-0v=\lambda v.$$  Thus $Y - \f{1}{n^2}J_{n^2}$ has the same spectrum as $Y$ except that one eigenvalue (the eigenvalue $1$ corresponding to eigenvector $e^{(n^2)}$) is shifted down by 1 (to eigenvalue $0$).  Consequently all eigenvalues of $Y - \f{1}{n^2}J_{n^2}$ are nonnegative, and $Y - \f{1}{n^2}J_{n^2}\succeq 0.$  
\hfill 
\end{proof}

\section{Conclusions}\label{sec:conc}
In this paper, we introduced simplicial TSP instances to show that the integrality gap of an SDP from  de Klerk and Sotirov \cite{Klerk12b}  is   unbounded, and moreover, nonmonotonic.  The simplicial TSP instances imply the unbounded integrality gap of every SDP relaxation of the TSP mentioned in the survey in Section 2 of Sotirov \cite{Sot12}, as well as the unbounded integrality gap of an SDP from Anstreicher \cite{Ans00}.  The simplicial instances  thus form a litmus test for new SDP relaxations of the TSP and motivate two questions.

\begin{ques}
Find an SDP relaxation of the TSP with finite integrality gap (without directly adding in the subtour elimination constraints of the subtour LP).
\end{ques}
It would suffice, for example, to find SDP constraints that implied scaled solutions lie in the Minimum Spanning Tree polytope.

\begin{ques}\label{q2}
What are the integrality gaps of any of the TSP SDP relaxations when the subtour elimination constraints are added?  Can any be shown to beat $\f{3}{2}$?
\end{ques}
To our knowledge, the only SDP for which Question \ref{q2} has been answered is for the SDP of  Cvetkovi{\'c} et al.\ \cite{Cvet99}: Goemans and Rendl \cite{Goe00} show that any feasible solution to the subtour LP gives an equivalent feasible to their SDP of the same cost; adding the subtour elimination constraints to this SDP thus effectively is the same as just solving the subtour LP.

 \section*{Acknowledgments}
This material is also based upon work supported by the National Science Foundation Graduate 
Research Fellowship Program under Grant No. DGE-1650441. Any opinions,
findings, and conclusions or recommendations expressed in this material are those of the
authors and do not necessarily reflect the views of the National Science Foundation.

\bibliography{bibliog} 
\bibliographystyle{abbrv}

\clearpage
\appendix
\section{Proofs of Trigonometric and Algebraic Identities}

In the appendix, we sketch pertinent results from Gutekunst and Williamson \cite{Gut17}.

\begin{lm*}[Lemma \ref{lem:trig}]
Let $n$ be even and $0<k<n$ be an integer.  Then $$\sum_{j=1}^d \cos\left(\f{2\pi jk}{n}\right) = \f{-1+(-1)^k}{2}.$$  
\end{lm*}
\begin{proof}
Our identity is a consequence of Lagrange's trigonometric identity (see, e.g., Identity 14 in Section 2.4.1.6 of Jeffrey and Dai \cite{Jeff08}), which states, for $0<\th<2\pi,$ that $$\sum_{j=1}^m \cos(j \theta) = -\f{1}{2} + \f{\sin\left(\left(m+\f{1}{2}\right)\theta\right)}{2\sin\left(\f{\theta}{2}\right)}.$$  Taking $\theta = \f{2\pi k}{n}$ and using $n=2d$, we obtain:
\begin{align*}
\sum_{j=1}^d \cos\left(\f{2\pi k}{n} j\right) &= -\f{1}{2} + \f{\sin\left( \pi k+ \f{\pi k}{n}\right)}{2\sin \f{\pi k}{n}} \\ 
&=- \f{1}{2} + (-1)^k \f{1}{2},
\end{align*}
where we recall that $\sin(\pi+\theta)=-\sin(\theta).$ \hfill
\hfill \end{proof}
\noindent Notice that when $k=0$ or $k=n$, the sum is $d$.

\begin{prop*}[Proposition \ref{prop:recall}]
\ \begin{enumerate}
\item $\sum_{i=1}^d a_i  = \sum_{i=1}^d b_i = 1.$  Equivalently, $a^{(0)}=b^{(0)}=1.$
\item $b^{(k)} =-\left(1-\f{2}{n}\right)a^{(k)} - \f{2}{n}.$
\item For $k=1, ..., d,$ $$a^{(k)} = \begin{cases} \f{d-2}{n-2}, & \text{ if } k=1 \\ -\f{2}{n-2}, & \text{ otherwise}. \end{cases}$$ 
\item $b_1 \leq \f{4\pi^2}{n^3}$
\end{enumerate}
\end{prop*}

\begin{proof}
For the first result, we use the identity
 $$\sum_{j=1}^d \cos\left(\f{2\pi jk}{n}\right) = \f{-1+(-1)^k}{2}$$   with $k=1$. 
 Then  
 \begin{align*}
\sum_{i=1}^d a_i &= \f{2}{n-2} \sum_{i=1}^d  \left(\cos\left(\f{\pi i}{d}\right)+1\right)= \f{2}{n-2}\left(-1 + d\right)= 1.
\end{align*} 
Similarly

\begin{align*}
\sum_{i=1}^d b_i &= \sum_{i=1}^{d-1} \left(\f{4}{n}-\left(1-\f{2}{n}\right)a_i\right) + \left(\f{2}{n}-\left(1-\f{2}{n}\right)a_i\right) = (d-1)\f{4}{n} + \f{2}{n} -\left(1-\f{2}{n}\right)\sum_{i=1}^d a_i= 1.
\end{align*} \hfill

For the second result,
\begin{align*}
b^{(k)} &=\sum_{i=1}^d    \cos\left(\f{2\pi ik}{n}\right)  b_i  \\
&=\left(\sum_{i=1}^{d-1}    \cos\left(\f{2\pi ik}{n}\right)  \left(\f{4}{n}-\left(1-\f{2}{n}\right)a_i\right)\right)+\cos\left(\f{2\pi dk}{n}\right)  \left(\f{2}{n}-\left(1-\f{2}{n}\right)a_d\right) \\
&=\f{4}{n} \left(\sum_{i=1}^{d}    \cos\left(\f{2\pi ik}{n}\right) \right)-\left(1-\f{2}{n}\right) \left(\sum_{i=1}^{d}    \cos\left(\f{2\pi ik}{n}\right) a_i\right)-\cos(\pi k)   \left(\f{2}{n}\right).
\intertext{Using Lemma \ref{lem:trig}:}
&=\f{4}{n}  \left(\f{-1+(-1)^k}{2}\right)-\left(1-\f{2}{n}\right)a^{(k)}-\left(-1\right)^k  \left(\f{2}{n}\right)\\
&=-\left(1-\f{2}{n}\right)a^{(k)} - \f{2}{n}.
\end{align*}

For the third result, we use the product-to-sum identity for cosines and then do casework using Lemma \ref{lem:trig}. We have:
\begin{align*}
a^{(k)} &=\sum_{i=1}^d  \cos\left(\f{2\pi i k}{n}\right) a_i\\
 &=\f{2}{n-2}  \sum_{i=1}^d \left( \cos\left(\f{2\pi ik}{n}\right) +\cos\left(\f{2\pi ik}{n}\right) \cos\left(\f{\pi i}{d}\right) \right)\\
 &=\f{2}{n-2}  \sum_{i=1}^d \left( \cos\left(\f{2\pi ik}{n}\right) + \f{1}{2}\cos \left(\f{2\pi i(k+1)}{n} \right)  +  \f{1}{2}\cos \left( \f{2\pi i(k-1)}{n}\right) \right).
\intertext{We  cannot apply Lagrange's trigonometric identity only when $k=1$, so that}
&= \begin{cases} \f{2}{n-2}  \left( \f{-1+(-1)^k}{2} + \f{-1 + (-1)^{k+1}}{4}  + \f{-1 + (-1)^{k-1}}{4}\right), & \text{ if } k>1 \\
\f{2}{n-2}  \left( -1+ 0+ \f{1}{2} d\right), & \text{ if } k=1 \end{cases} \\
&= \begin{cases} -\f{2}{n-2}, & \text{ if } k>1 \\
\f{d-2}{n-2} , & \text{ if } k=1. \end{cases} 
\end{align*} \hfill

Finally, the fourth result follows from Taylor series with remainder, 
 $$\cos\left(\f{\pi}{d}\right) = 1-\f{\pi^2}{2d^2} + \f{1}{4!}\f{\pi^4}{d^4} \cos\left(\xi_{1/d}\right) \geq 1-\f{\pi^2}{2d^2}, 
$$ 
where $\xi_{1/d} \in [0, \f{1}{d}].$  Hence $$b_1 = \f{2}{n}\left(1-\cos\left(\f{\pi}{d}\right)\right) \leq \f{2}{n} \f{\pi^2}{2d^2}=\f{4 \pi^2}{n^3}.$$ \hfill
\end{proof}

\end{document}